%% file: main.tex
\documentclass[11pt]{article}
\usepackage[letterpaper, margin=1in]{geometry} 
\usepackage{lastpage}
\usepackage{fancyhdr}
\def\randomcr{0}

\usepackage[utf8]{inputenc}
\usepackage[T1]{fontenc}
\usepackage{lmodern} %
\usepackage{microtype} %
\usepackage{setspace} %
\usepackage{booktabs}
\usepackage{tikz}
\usepackage{pgfplots}
\pgfplotsset{compat=1.18}
\usepackage{algorithm}
\usepackage[noend]{algpseudocode}
\usepackage{amsmath, amssymb, amsthm}
\renewcommand{\geq}{\geqslant} %
\renewcommand{\leq}{\leqslant}
\newcommand{\EET}{\ensuremath{\mathsf{EET}}}
\newcommand{\EIT}{\ensuremath{\mathsf{EIT}}}
\newcommand{\MI}{\ensuremath{\mathsf{MI}}}

\usepackage{thmtools}
\usepackage{thm-restate}
\usepackage{mathtools}
\usepackage{bm} %

\usepackage{graphicx}
\definecolor{darkblue}{rgb}{0.0, 0.0, 0.55}

\usepackage[style=alphabetic, maxalphanames=4, minalphanames=3, backref=true, natbib=true]{biblatex}
\usepackage{hyperref}
\hypersetup{
    colorlinks=true,
    linkcolor=darkblue,
    citecolor=darkblue,
    urlcolor=darkblue,
    pdftitle={Entropy Equivalence Testing}
}

\usepackage[capitalize]{cleveref}
\usepackage{framed}

\usepackage{titlesec}
\titleformat{\section}{\large\bfseries\sffamily}{\thesection}{1em}{}
\titleformat{\subsection}{\normalsize\bfseries\sffamily}{\thesubsection}{1em}{}

\newtheorem{theorem}{Theorem}

\newtheorem{lemma}[theorem]{Lemma}
\newtheorem{corollary}[theorem]{Corollary}
\theoremstyle{definition}
\newtheorem{definition}[theorem]{Definition}
\newtheorem{remark}[theorem]{Remark}

\usepackage{texmacs}

\title{\textbf{\sffamily Entropy Equivalence Testing}}
\date{}
\author{\textbf{Clément L. Canonne}\\
        University of Sydney\\
        \texttt{clement.canonne@sydney.edu.au}
        \and
        \textbf{Yash Pote}\\
        National University of Singapore\\
        \texttt{yashppote@gmail.com}\\
        \and
        \textbf{Jonathan Scarlett}\\
        National University of Singapore\\
        \texttt{scarlett@comp.nus.edu.sg}
        \and
        \textbf{Joy Qiping Yang}\\
        University of Sydney\\
        \texttt{qiping.yang@sydney.edu.au}}

\begin{document}
\maketitle

\begin{abstract}
    We introduce the problem of \emph{entropy equivalence testing} for probability distributions, a relaxation of the well-studied closeness testing problem, where the distribution testing algorithm is now only required to distinguish, given samples from two unknown distributions $p,q$ and a parameter $\varepsilon \in(0,1/2]$, between $p=q$ and $|H(p)-H(q)| \geq \varepsilon$ (where $H$ denotes the Shannon entropy). We provide a time- and sample-efficient algorithm for this task, showing that the optimal sample complexity for this task can be significantly lower than that of closeness testing.  
    As an application, we leverage this result to provide the first non-trivial testing algorithm for (standard) closeness of low-degree \emph{Bayesian networks}, which significantly improves on either the sample or time complexity of a baseline based on full learning.
\end{abstract}
\hrule
\vspace{1em}

\input{introduction}
\input{equivalence}

\section*{Acknowledgements}
This work is supported by the National University of Singapore under the Presidential Young Professorship scheme. JY is supported by a JD Technology Research Scholarship in Artificial intelligence. We would like to thank Sayantan Sen, Jan Seyfried and Marco Tomamichel for the discussion on their work of testing Conditional Mutual Information. JY would also like to thank Davin Choo for the early discussions predating this project. This work is conducted while JY was interning at National Institute of Informatics, hosted by Yuichi Yoshida. We thank the reviewers of RANDOM 2026, whose comments help improve the presentation of this work.

\printbibliography
\input{Appendix}
\end{document}

%% file: introduction.tex
\section{Introduction}
Distribution testing, a subfield of property testing~\cite{RubinfeldS96,GoldreichGR98} concerned with efficient algorithms for deciding statistical properties of data, has witnessed a significant interest since its systematic introduction by Batu, Fortnow, Rubinfeld, Smith, and White~\cite{BatuFRSW00}, with a flurry of works over the past two decades addressing a range of hypothesis testing tasks, under various models of distribution testing (see, e.g.,~\cite{Rubinfeld12,Canonne20,Canonne22} and~\cite[Chapter~11]{Goldreich17} for surveys and expositions).

Two of the flagship and most fundamental tasks in distribution testing are \emph{identity testing} (and its special case, \emph{uniformity testing}), and \emph{closeness testing} (also known as \emph{equivalence} testing): in the first, an algorithm is given independent samples from an (unknown) discrete probability distribution, and must decide whether the distribution is equal to some known reference distribution, or significantly deviates from it in a meaningful notion of distance (typically total variation distance). In the second, the algorithm has access to samples from \emph{two} unknown distributions, and similarly must decide whether they are equal, or significantly far from each other. Here, ``significantly far'' is quantified by an input parameter $\varepsilon>0$, and the algorithm is required to be correct with high probability $1-\delta$,\footnote{One typically sets $\delta$ to be a small constant, e.g., $\delta=1/10$: this is without much loss of generality, as one can amplify this by repetition to any target $1-\delta$, at the cost of an $O(\log(1/\delta))$ factor in sample and time complexity.\label{ft:delta}} while minimizing the \emph{sample complexity}, that is, the total number of samples taken from the distribution(s) in the worst case as a function of the domain size $n$, distance parameter $\varepsilon$, and error probability $\delta$.

The sample complexity of these two flagship tasks is by now well understood, with several time-efficient algorithms known to achieve the optimal sample complexity, as a function of all parameters: namely,
\[
\Theta\left(\frac{\sqrt{n\log(1/\delta)}+\log(1/\delta)}{\varepsilon^2}\right)
\]
for identity testing~\cite{Paninski08,AcharyaDK15,DK16,DiakonikolasGPP18}, and
\[
\Theta\left(\frac{n^{2/3}\log^{1/3}(1/\delta)+\log(1/\delta)}{\varepsilon^{4/3}}+\frac{\sqrt{n\log(1/\delta)}+\log(1/\delta)}{\varepsilon^2}\right)
\]
for closeness testing~\cite{BatuFRSW00,CDVV14,DK16,DiakonikolasGKP21}. %
However, in spite of these impressive advances, these sample complexities can remain prohibitively large for many combinations of $(n,\varepsilon)$, notably for small $\varepsilon$ corresponding to distributions that may be very close to one another.
Given that their sample complexities are constant-factor \emph{optimal}, there would seem to be little room for improvement, yet two possible avenues remain. The first is to \emph{relax the guarantees}: total variation is known to be a very stringent measure of distance, and testing with respect to a more permissive distance measure may enable significant sample complexity savings, while still being ``good enough'' for the application at hand. This direction was partially explored in, for instance,~\cite{DaskalakisKW18}. The second avenue, particularly well suited to high-dimensional data, is to \emph{make additional structural assumptions} on the distribution: high-dimensional data is seldom fully unstructured, and one can usually leverage domain knowledge to make modelling assumptions, for instance that the probability distributions exhibit a low-degree Bayesian network structure, or are consistent with an Ising model. This line of work, in the context of distribution testing, was initiated in~\cite{CDKS20,DaskalakisDK18}.

Our work makes progress on both of these fronts.  We first define a relaxation of  closeness testing, which we term \emph{entropy equivalence testing}, where one is asked to detect if the two unknown distributions have significantly different (Shannon) entropy:\footnote{Recall that the (Shannon) entropy of a probability distribution $p$ over a discrete domain $\mathcal{X}$ is defined as $H(p) = \sum_{x\in \mathcal{X}} p(x)\log\frac{1}{p(x)}$.}

\begin{framed}
\noindent\textbf{Entropy Equivalence Testing (\EET)}:
Given sample access to distributions $p,q$ over a known domain of size $n$ and parameters $\varepsilon \in(0,1/2]$ and $\delta\in(0,1]$, distinguish, with probability
 at least $1-\delta$, between the following two cases: (1)~$p=q$, and (2)~$|H(p) - H(q)| \ge \varepsilon$.
\end{framed}
We note that the ``single-distribution'' version of this task, when $q$ is fully known to the algorithm, was recently introduced~\cite{CY24}, under the name Entropy Identity Testing (\EIT). \cite{CY24} then leveraged their algorithm for \EIT~to obtain an improved identity testing algorithm for ``standard'' identity testing (i.e., in total variation distance) of \emph{bounded-in-degree Bayes nets}, establishing a connection between the two aforementioned avenues of research. We build upon and deepen this connection to obtain, from our algorithm for \EET, the first non-trivial \emph{closeness} testing algorithm for Bayes nets.

\subsection{Our results}

We now describe our results in more detail. Our first contribution is an efficient algorithm for the \EET~ task defined above (where we implicitly set the probability of failure $\delta$ to be a small constant):
\begin{theorem}\label{theo:eet:intro}
    There exists an algorithm for the Entropy Equivalence Testing problem with sample complexity
    \[
        \tilde{O} \!\left( \frac{n^{3 / 4}}{\varepsilon}\right) + O\left(\frac{\log^2(n/\varepsilon)}{\varepsilon^2} \right)
    \]
    and running time near-linear in the number of samples.
\end{theorem}
Combining \cref{theo:eet:intro} and the ``baseline tester'' discussed below (see~\cref{ssec:overview}) which follows from a reduction from entropy difference to total variation distance and has sample complexity $\tilde{O}\left( \max \left( \frac{n^{2 / 3}}{\varepsilon^{4 / 3}},
\frac{\sqrt{n}}{\varepsilon^2} \right) \right)$, one gets, after simplifying the resulting sample complexity, our main result:

\begin{corollary}[Main testing result]\label{theo:eet:full_informal}
    There exists an algorithm for the Entropy Equivalence Testing problem with sample complexity
    \[ \tilde{O} \left( \min \left\{
   \frac{n^{3 / 4}}{\varepsilon}, \frac{n^{2 / 3}}{\varepsilon^{4 / 3}} \right\}\right) + O\left(\frac{\log^2(n/\varepsilon)}{\varepsilon^2} \right)
   \]
   and running time near-linear in the number of samples.
\end{corollary}
We complement this positive result with an information-theoretic lower bound, showing that our algorithm has nearly optimal sample complexity:
\begin{restatable}[Lower bound]{theorem}{lowerbound}
\label{thm:eet-lb:intro}
Any algorithm for the Entropy Equivalence Testing problem must draw
\begin{align*}
 \widetilde\Omega\!\left(\min\!\left\{\frac{n^{3/4}}{\varepsilon},\,
    \frac{n^{2/3}}{\varepsilon^{4/3}}\right\}\right)
  + \Omega\!\left(\frac{\log^2 n}{\varepsilon^2}\right)
\end{align*}
total samples from $p$ and $q$.
\end{restatable}

Finally, we show how to leverage our algorithm for \EET~to obtain a closeness testing algorithm for bounded in-degree Bayes nets: we emphasize that, to the best of our knowledge, this is the \emph{first} testing algorithm for this problem, besides the ``na\"ive'' (and computationally inefficient) testing-by-learning approach that fully learns the two unknown distributions with $\tilde{O}(2^dn/\varepsilon^2)$ samples, in time $2^{\Omega(n)}$:
\begin{restatable}{theorem}{bayesnet}
    \label{theo:bayesnet:closeness:intro}
    There exists an algorithm for closeness testing (in total variation distance) of in-degree $d$ Bayesian networks with sample complexity
\[
s = \tilde{O} \left( \min \left\{ \frac{2^{3 d / 4} n}{\varepsilon^2},
   \frac{2^{2 d / 3} n^{4 / 3}}{\varepsilon^{8 / 3}} \right\} +
   \frac{n^2}{\varepsilon^4} \right) , 
\]
and running time $\tilde{O}(n^{d+1} s)$.
\end{restatable}

Interestingly, the techniques underlying~\cref{theo:bayesnet:closeness:intro} can also be easily adapted to obtain an algorithm for identity testing of Bayes nets, improving upon the work of~\cite{CY24}: see~\cref{cor:improve:upon:cy24}.

%% file: equivalence.tex
\subsection{Technical overview}
We now provide an overview of the techniques and main ideas underlying our results.
\label{ssec:overview}
\subsubsection{Upper bound for \EET}
A bound of $\tilde{O} \left( \max \left\{ \frac{n^{2 / 3}}{\varepsilon^{4 / 3}},
\frac{\sqrt{n}}{\varepsilon^2} \right\} \right)$ readily follows from a standard inequality, allowing one to reduce the entropy difference problem to closeness testing in TV: namely, when $d_{\tmop{TV}} (p, q) \leqslant \frac{1}{2}$, then $| H (p) - H (q) |
\leqslant d_{\tmop{TV}} (p, q) \log \frac{n}{d_{\tmop{TV}} (p, q)}$ \cite[Lemma 2.7]{CK11}. Another straightforward upper bound is obtained by estimating the entropy, a harder task, which takes $\Theta \left( \frac{n}{\varepsilon \cdot \log n} + \frac{\log^2
n}{\varepsilon^2} \right)$ samples \cite{WY16}. This hints at the possibility of getting a better bound of $O(n^c/\varepsilon + \log^2 n/\varepsilon^2)$ for some absolute constant $c<1$.
Surprisingly, the final upper bound of $\tilde{O} \left( \min \left\{ \frac{n^{2 / 3}}{\varepsilon^{4 / 3}},
\frac{n^{3 / 4}}{\varepsilon} \right\}\! +\! \frac{1}{\varepsilon^2} \right)$ is related to the squared Hellinger distance (and not the total variation). This is a result of considering two testers that can solve \EET~independently, and choosing the one with lower sample complexity. The main contribution of our work is providing an alternative tester to the TV reduction tester above.

\paragraph*{An $\tilde{O}(\frac{n^{3/4}}{\varepsilon}+\frac{1}{\varepsilon^2})$ tester.} Our analysis starts with a reduction to testing both farness (vs.~equality) in Hellinger distance and a cross-entropy difference term (with respect to the uniform mixture) via some algebraic manipulations (see \Cref{lemma:reduce_to_Hellinger_or_entropy}). The problem thus becomes testing whether Hellinger distance is large (which we know costs $O(\min\{\frac{n^{3/4}}{\varepsilon},\frac{n^{2/3}}{\varepsilon^{4/3}}\}) \leqslant O(\frac{n^{3/4}}{\varepsilon})$) %
and testing whether the cross-entropy term is large (since one of them will be large if the entropy difference is large). A natural idea is of course to employ a plug-in estimator for the cross entropy difference (see~\eqref{eq:cross_entropy_expression}). Unfortunately, this empirical estimator has an inherent\footnote{In the works of \citet{WY16}, their plug-in estimator also exhibits similar bias in the analysis, and they resort to machinery from polynomial approximation to reduce the bias (which will increase the variance the higher the degree of the polynomial).} bias of the form $\sum_{i \in [n]} \frac{| p_i - q_i |}{m (p_i + q_i)}$, where $m$ is the (equal) number of samples we take from $p$ and $q$ to compute the estimate.
Our argument to deal with the bias is as follows: If the expectation of the plug-in estimator has \emph{low bias} (at most $O(\varepsilon)$), then we can rely on its estimate alone. On the other hand, if it has \emph{high bias} (at least $\Omega(\varepsilon)$), we hope to catch it and reject (i.e., declare that $|H(p)-H(q)| > \epsilon$) early.

\paragraph*{Low mass regime: Reduction to TV.} One main challenge to designing this test is that the bias is impossible to test without a lower bound on $p_i+q_i$ (the variance of estimating the bias could be huge). The key to passing this hurdle is an idea inspired by a different algorithm from \citet{DK16}, whose analysis implies that for the elements $i\in[n]$ where $q_i$ (or $p_i$) is small, we can perform tests more efficiently (in TV) via access to its conditional distribution (as efficiently as $\tilde{O}({n^{3/4}}/{\varepsilon})$), with the intuition being: The TV distance between the conditional distributions is actually larger (inverse of its mass size $\tilde{\Omega}({\varepsilon}/{q(T)})$, where $T$ is the set of small-mass elements). Via a reduction to TV from cross-entropy difference, we can also test farness on these elements with $\tilde{O}({n^{3/4}}/{\varepsilon})$ samples.

\paragraph*{Detecting bias in the high mass regime.} Capitalizing on this, we will consider exclusively elements with large mass: $p_i \geqslant \tilde{\Omega}( {\varepsilon}/{n^{3/4}})$. With a bound on $m (p_i+q_i) \geqslant \Omega(1)$ (since we take $m=\tilde{O}(\frac{n^{3/4}}{\varepsilon})$), this reduces to testing farness in TV distance, which can be performed with $O(\sqrt{n}/\varepsilon^2)$ samples (to bound the bias term by $O(\varepsilon)$), rather than the $O(n^c/\varepsilon)$ for some $c<1$ that we set out for. Luckily our problem can be easier (depending on relation between $n \text{ and }\varepsilon$) than testing in TV (see \Cref{lemma:check_if_Z_has_small_bias}). Intuitively, this is related to testing in the instance-optimal setting \cite{ValiantV14,DK16}~--~one can test with sample complexity as a function of $p, q$ (for our case $p,q$ has exclusively large probability masses), rather than as a function of the domain size (the worst case over all possible $p, q$). Interestingly, to detect the bias, we are using the same test statistic (with different threshold) first used in \citet{CDVV14} to test in TV distance and subsequently in \citet{DaskalakisKW18} to test in Hellinger distance.

Had we not made a `reject' decision after the bias check in the high mass regime (nor rejected due to its low mass cross-entropy difference), it would suggest that the bias on the high mass regime is at most $O(\varepsilon)$ and therefore, the estimator for $\sum_i (p_i - q_i) \log((p_i+q_i)/2)$ is reliable: A standard analysis of expectation and variance followed by an application of Chebyshev's inequality then yields \Cref{theo:eet:intro}.

\subsubsection{Lower bound for \EET}
The lower bound is a combination of $\tilde{\Omega}(\min\{\frac{n^{2/3}}{\varepsilon^{4/3}},\frac{n^{3/4}}{\varepsilon}\})$ and $\Omega(\frac{\log^2 n}{\varepsilon^2})$. The former is derived via a reduction from testing mutual information to entropy difference. The task of testing mutual information (\MI) is as follows: if $I(X:Y)=0$, then accept with probability $0.99$; If $I(X:Y) \geqslant \varepsilon$, then reject with probability $0.99$. Given the close relation between mutual information and entropy: $I(X:Y) = H(X) + H(Y) - H(X, Y)$, one can reduce the problem of \MI~testing to \EET~by treating $p$ as the joint distribution of $X,Y$ and $q$ as the product of $X$ and $Y$ (we can simulate sample access to the product of $X$ and $Y$ with two samples from the joint distribution). Therefore, the existing lower bound from \cite[Corollary 6.2]{SST25} holds for our problem. The latter lower bound of $\Omega(\frac{\log^2 n}{\varepsilon^2})$ comes from \cite[Lemma 2.7]{CY24}, given the reduction from identity testing to equivalence testing.

\subsubsection{Application to equivalence testing of Bayes nets}
As laid out by \citet{CY24}, one can utilize testing entropy difference as a subroutine to testing Bayes nets. The key of the analysis is derived via the celebrated Chow-Liu decomposition~\cite{ChowL68}: If two Bayes nets $p$ and $q$ have similar local entropies, then they must share approximately similar structures (their corresponding DAGs are close in some sense; see \cite[Lemma 3.4]{CY24}).
\ifnum\randomcr=1
Once we confirm that they share a similar structure, the problem boils down to checking local (every subset of random variables of size at most $d+1$) KL-divergences.
\else
Once we confirm that they share a similar structure, the problem boils down to checking local (every subset of random variables of size at most $d+1$) KL-divergences (see \eqref{eq:local_KL_divergence_large}).
\fi
While KL-divergences are known to be untestable \cite[Theorem 7]{DaskalakisKW18} in the extreme cases when $q(x)=0$, there is hope if one can obtain a lower bound on $q(x)$. In fact, we can reduce testing in KL to testing in Hellinger distance with an $O \left( \underset{x}{\max} \log \frac{p(x)}{q(x)} \right)$ factor overhead using \cite[Lemma 3.1]{SST25}. To get a lower bound on $q(x)$, we use a standard technique of testing the mixture of $q$ and the uniform distribution.

\subsection{Related works}
\paragraph*{Estimating entropy.} Entropy estimation is a well-studied problem in both the information theory \cite{WuY15,WY16,JiaoVHW15} and property testing \cite{ValiantV11,BatuDKR05} communities, with the near-optimal sample complexity $\Theta(\frac{n}{\varepsilon \log n} + \frac{\log^2 n}{\varepsilon^2})$ being eventually established by \citet{WY16,JiaoVHW15} (for discrete distributions). As we are considering the testing problem, our upper bound is less costly compared to its learning (estimation) variant. Moreover, other variants of entropy have been studied, such as Rényi entropy \cite{AcharyaOST15}. The more general problem of estimating functionals of distributions was also considered by \citet{JiaoVHW15}.

\paragraph*{Identity testing and closeness testing.}
The standard formulation of identity testing asks the following problem: Given a full description of a hypothesis $q$, how many i.i.d.~samples are needed from $p$ to distinguish between $p=q$ and $d_{\tmop{TV}}(p,q)\geqslant \varepsilon$, where $p$ and $q$ are both distributions over support of size $[n]$? This problem is well-known to have a sample complexity of $\Theta\left({\sqrt{n}}/{\varepsilon^2}\right)$, a much lower cost compared to fully learning every parameter of $p$, which takes $\Theta({n}/{\varepsilon^2})$ samples. A similar problem can be asked if one is instead only provided sample access to $p$ and $q$, both unknown. This is the closeness testing problem (in TV) and it is known that a $\Theta \left(\max\{\frac{\sqrt{n}}{\varepsilon^2},\frac{n^{2/3}}{\varepsilon^{4/3}}\} \right)$ sample complexity is necessary and sufficient, making this a slightly harder problem than identity testing while still saving on samples compared to fully learning $p$ and $q$. Our problem differs from the standard formulation by testing in the entropy difference rather than in TV distance.
Most relevant to our work is the work of \citet{CY24}, which solves the problem of entropy identity testing (up to polylogarithmic factors) with a sample complexity of $\tilde{\Theta}(\frac{\sqrt{n}}{\varepsilon} + \frac{1}{\varepsilon^2})$. Another relevant work is that of \citet{DaskalakisKW18}, which asks the question of testing in different distances, e.g., distinguishing between $d_{\tmop{KL}}(p \| q) \leqslant \varepsilon^2/100$ and $d_{\tmop{H}}(p,q) \geqslant \varepsilon$ for identity testing (sample access to $p$ and the full description of $q$). Interestingly, as shown in the same work, equivalence testing in squared Hellinger distance, i.e., distinguishing between $p=q$ and $d^2_H(p,q) \geqslant \varepsilon$ has a similar cost of $\Theta\left(\min\{\frac{n^{2/3}}{\varepsilon^{4/3}}, \frac{n^{3/4}}{\varepsilon}\}\right)$ (though perhaps not so surprising, as we are reducing to testing in squared Hellinger square in one of the cases we handle).
Our work adds to their results for equivalence testing in entropy difference.

\paragraph*{Testing mutual information.} Testing mutual information (\MI) asks the following question: Given sample access to a (discrete) distribution $p$ over two random variables $X, Y$ with a domain of size $|\Sigma_X|$ and $|\Sigma_Y|$ respectively, how many samples do we need from $p$ to distinguish between $I(X:Y)=0$ and $I(X:Y)\geqslant\varepsilon$?
The study of this problem was initiated by \citet{BhattacharyyaGPTV23} to solve the problem of learning a special variant of Bayes net (trees). While their \MI~tester is sufficient for learning trees-structured distributions near-optimally, their upper bound is loose in terms of dependence on $|\Sigma_X|$ and $|\Sigma_Y|$: $O(\frac{|\Sigma_X|\cdot|\Sigma_Y|}{\varepsilon})$. More recently, \citet{SST25} gave a near-optimal bound of $\tilde{\Theta}(\min \{\frac{|\Sigma_X|^{3/4}|\Sigma_Y|^{1/4}}{\varepsilon},\frac{|\Sigma_X|^{2/3}|\Sigma_Y|^{1/3}}{\varepsilon^{4/3}}\})$ (assuming $|\Sigma_X|\geqslant|\Sigma_Y|$). We note that this bound bears similarity (e.g., setting $|\Sigma_Y|=O(1)$) to ours due to the close relation between \MI~and entropy: $I(X:Y) = H(X) + H(Y) - H(X, Y)$. Indeed, \MI~testing can be reduced to \EET~naturally (as outlined above).
Before these works, \citet{BatuFFKRW01,CanonneDKS18} studied the problem of (conditional) independence testing in the TV distance.

\paragraph*{Testing Bayesian networks.} A Bayesian network (Bayes net) is a high-dimensional distribution over $n$ random variables that factorizes into $n$ conditional probabilities: $p(X_1,\ldots,X_n)=\prod_{i=1}^n p(X_i \mid \Pi_i)$, where $\Pi_i \in [n]$ is called the parent set of $X_i$. This factorization is associated with a directed acyclic graph (DAG) $G=(V,E)$ ($V=\{X_1,\ldots,X_n\}$), and if $X_j \in \Pi_i$, then the edge $(X_j, X_i) \in E$ (with $X_j$ being one of $X_i$'s parents). A distribution $p$ is called a Bayes net if its joint probability density function factorizes according to $G$ and $p$ is called Markov with respect to $G$.
Learning Bayesian networks from samples is a fundamental problem in learning theory, and there is a plethora of evidence that this problem is hard computationally \cite{ChickeringHM04,Dasgupta99,0001CGM25}.  On the other hand, time-efficient algorithms (polynomial in $n$, the number of random variables) do exist in some special cases \cite{ChooY0C24,BhattacharyyaGGSV25,BhattacharyyaGPTV23}.

Identity testing and closeness testing of Bayes nets were systematically studied by \citet{CDKS20}, though their stronger upper bounds only consider relaxed version of Bayes nets where both $p$ and $q$ are required be \emph{balanced} and \emph{non-degenerate}. While they  also provide a more general upper bound (without assuming balancedness and non-degeneracy for Bayes nets), they still require that the node ordering be the same between $p$ and $q$, a non-trivial assumption. Our upper bound does not make any assumption other than the bounded in-degree $d$, which is a natural and necessary parameter to consider, since testing unbounded degree Bayes nets is exponentially hard: $\Omega(\frac{2^{n/2}}{\varepsilon^2})$ (e.g., see \cite[Theorem B.1]{BCY22}). The strongest upper bound that we are aware of is by \citet{DaskalakisP17}, which can solve the identity testing\footnote{Note that their setting of identity testing actually corresponds to equivalence testing in our setting.} problem in $\tilde{O}(\frac{2^{d/2} n}{\varepsilon^2})$ samples,\footnote{Implicitly, using the Hellinger testing result of \cite[Theorem 1]{DaskalakisKW18}.} under the assumption that $p$ and $q$ have the same structure (Markov with respect to the same DAG). For equivalence testing, under the same assumption, they give a bound of $\tilde{O}(\min\{\frac{2^{3d/4} n}{\varepsilon^2}, \frac{2^{2d/3} n^{4/3}}{\varepsilon^{8/3}} \})$.\footnote{Implicitly, using the optimal Hellinger equivalence test \cite{DaskalakisKW18}.} Our results match their bounds if $d\gg\log (n/\varepsilon)$, and our algorithm will work without the extra assumption that $p$ and $q$ share the same structure. 
A lower bound of $\Omega(\frac{2^{d/2}n}{\varepsilon^2})$ when $d \leqslant \log n$~\cite[Theorem 4.1]{BCY22} holds for identity testing of Bayes nets and thus also for its closeness testing variant (our setting).

\subsection{Preliminaries}
Define the factorial moment by $(x)_m \assign \frac{x!}{(x - m) !}$, which
gives an unbiased estimator for the monomials of the Poisson mean: namely, if $X\sim\operatorname{Poi}(\lambda)$, then $\mathbb{E}
[(X)_m] = \lambda^m$ for all $m \geq 0$. More generally, we have the following fact, whose proof we provide in Appendix~\ref{app:deferred} for completeness:

\begin{restatable}{fact}{factorialmoment}\label{fact:factorial_moment_functional_shifting}
Let $X$ be a Poisson random variable with parameter $\lambda$, and $f\colon\mathbb{R}\to\mathbb{R}$ be any function, then, we have that $\mathbb{E}
[(X)_m \cdot f (X)] = \lambda^m \cdot \mathbb{E} [f (X + m)]$. In particular, choosing $f\equiv 1$, we have that $\mathbb{E} [(X)_m] = \lambda^m$ for all $m\geq 0$.
\end{restatable}

\begin{restatable}[Folklore]{fact}{distinguish_small_vs_big}
\label{fact:distinguish_small_vs_big}
For any $0 < \alpha \leqslant 1/2$ and $\varepsilon \in (0,1]$, testing whether the bias of a coin is at most $\alpha$ or at least $\alpha (1 + \varepsilon)$, with probability at least $1-\delta$, can be done with $n=O\left(\frac{1}{\alpha \varepsilon^2}\log \frac{1}{\delta} \right)$ samples.
\end{restatable}

\paragraph*{Distances and norms.} We will use a number of standard distance measures for two distributions $p, q$ over $[n]$, and various known inequalities between them, e.g., see \cite{SasonV16}.  
The \emph{squared Hellinger distance} is
\begin{align*}
   d_H^2(p, q) := \frac{1}{2}\sum_{i \in [n]}\!\left(\sqrt{p_i} - \sqrt{q_i}\right)^2
\end{align*}
The \emph{total variation distance} is
$d_{\tmop{TV}}(p, q) := \frac{1}{2}\|p - q\|_1 = \frac{1}{2}\sum_{i \in [n]}|p_i - q_i|$.
The \emph{$\ell_2$ distance} is
$\|p - q\|_2 := \bigl(\sum_{i \in [n]}(p_i - q_i)^2\bigr)^{1/2}$,
and the \emph{$\chi^2$-divergence} of $p$ from $q$ is
$d_{\chi^2}(p, q) := \sum_{i \in [n]}(p_i - q_i)^2/q_i$.
These satisfy the chain of inequalities
\begin{align*}
   d_{\tmop{TV}}(p, q) \leq \sqrt{2}\, d_H(p, q)
   \qquad \text{and} \qquad
   d_H^2(p, q) \leq d_{\tmop{TV}}(p, q)
   \leq \sqrt{d_{\chi^2}(p, q)},
\end{align*}
and the equivalence $\frac{1}{2}\sum_{i}(p_i - q_i)^2/(p_i + q_i) = \Theta(d_H^2(p,q))$.
The \emph{Kullback--Leibler (KL) divergence} is $d_{\tmop{KL}}(p \| q) := \sum_{i}p_i \log(p_i/q_i)$,
which satisfies $d_{\tmop{KL}}(p\|q) \leq d_{\chi^2}(p, q)$ and
Pinsker's inequality $d_{\tmop{TV}}(p,q) \leq \sqrt{d_{\tmop{KL}}(p\|q)/2}$.

\paragraph*{Notation.} In what follows, for a given integer $n\geq 1$, we denote by $[n]$ the set $[n] = \{1,2,\dots,n\}$, and by $U_n$ the uniform distribution on $[n]$. When clear from context, we will omit the subscript, and write $U$ instead of $U_n$. Given a probability distribution $p$ over $[n]$ and a set $S\subseteq[n]$ such that $p(S) > 0$ we write $p(\cdot \mid S)$ for the conditional distribution of $p$ on $S$. For the high dimensional (Bayes net) section, we use $p_{T}$ to denote the marginalization of $p$ over $T$, where $T$ is a subset of $\{X_1, \ldots, X_n\}$, e.g., $p_{X_1,X_2}$ denotes the marginal distribution of $(X_1, X_2)$. We also use $\Pi_i$ to denote the parent set of the random variable $X_i$ in a Bayes net.

Throughout the paper, we use $\log$ for the natural logarithm, and adopt the standard asymptotic notation $O(\cdot),\Omega(\cdot),\Theta(\cdot)$, as well as the closely related $\tilde{O}(\cdot),\tilde{\Omega}(\cdot),\tilde{\Theta}(\cdot)$, which hide polylogarithmic factors in their argument.

\input{upperbound}
\input{lowerbound}

%% file: upperbound.tex
\section{An \texorpdfstring{
    $\tilde{O}(\frac{n^{3/4}}{\varepsilon} + \frac{\log^2(n/\varepsilon)}{\varepsilon^2})$}{
    O~(n\string^(3/4)/eps + log\string^2(n/eps)/eps\string^2)} Upper Bound for \EET}
In this section, we establish our main theorem,~\cref{theo:eet:intro}. We begin with a simple algorithmic ``trick'' that allows us to reduce the task to one where both $p,q$ put at least some minimum probability $\eta > 0$ on every element of the domain:
\begin{lemma}\label{lemma:mass_floor}
Let $p, q$ be distributions over $[n]$, and set $\eta := \varepsilon / (4 \log(n/\varepsilon))$.
Let $\tilde p := (1-\eta) p + \eta U$ and likewise for $\tilde q$. Then
$\tilde p_i, \tilde q_i \geq \varepsilon / (4 n \log(n/\varepsilon))$
for every $i \in [n]$, and
\begin{align*}
\left | |H(\tilde p) - H(\tilde q)| - |H(p) - H(q)| \right | \leq \varepsilon / 2.
\end{align*}
Moreover, one sample from $\tilde{p}$ (resp.~$\tilde{q}$) can be simulated from one sample from $p$ (resp. $q$).
\end{lemma}
\begin{proof}
    Since $d_{\tmop{TV}} (p, \tilde{p}) \leq \frac{\eta}{2} =
    \frac{\varepsilon}{8 \log (n / \varepsilon)}$, by invoking the bound on entropy difference with respect to TV distance in \cite[Lemma 2.7]{CK11}, 
    we get 
    \[
    | H (p) - H (\tilde{p}) | 
    \leq
    \frac{\varepsilon}{8 \log (n / \varepsilon)} \log \left( \frac{n \cdot 8
    \log (n / \varepsilon)}{\varepsilon} \right) = \varepsilon \cdot \frac{\log (n
    / \varepsilon) + \log (8 \log (n / \varepsilon))}{8 \log (n /
    \varepsilon)} \leq \varepsilon/4.
    \]
    Similarly, we have $| H (q) - H
    (\tilde{q}) | \leq \varepsilon/4$. By the triangle inequality, we then have
    \[ | H (\tilde{p}) - H (\tilde{q}) | \leq | H (p) - H (\tilde{p}) | + | H
       (q) - H (\tilde{q}) | + | H (p) - H (q) | \leq \varepsilon / 2 + | H (p) -
       H (q) | ,
    \]
    and
    \[ | H (p) - H (q) | \leq | H (\tilde{p}) - H (\tilde{q}) | + | H (p) - H
       (\tilde{p}) | + | H (q) - H (\tilde{q}) | \leq \varepsilon / 2 + | H
       (\tilde{p}) - H (\tilde{q}) |, \]
    establishing the lemma.
\end{proof}
Of course, if $p=q$, then $\tilde{p}=\tilde{q}$; \cref{lemma:mass_floor} guarantees that if $|H(p) - H(q)| \geq \varepsilon$, then $|H(\tilde p) - H(\tilde q)| \geq \Omega(\varepsilon)$. That is, $\tilde p, \tilde q$ satisfy the promise with the
same $\varepsilon$ (up to constants) whenever $p,q$ do, so we can henceforth assume that $\min_{i\in[n]} p_i, \min_{i\in[n]} q_i \geq
\frac{\varepsilon}{4 n \log (n / \varepsilon)}$.
With this in hand, we are ready to establish our main theorem, restated, as follows:
\begin{theorem}[\cref{theo:eet:intro}, restated]\label{theo:eet}
    There exists an algorithm (\cref{alg:eet}) which, given sample access to two unknown probability distributions $p,q$ over a known domain $[n]$, as well as parameter $\varepsilon \in (0,1/2]$, takes
    \[
        O \left( \frac{n^{3 / 4}\log(n/\varepsilon) \cdot \log n}{\varepsilon} + \frac{\log^2 (n /
   \varepsilon)}{\varepsilon^2} \right)
    \]
    samples from both $p$ and $q$, and has the following guarantees:
    \begin{itemize}
        \item If $p=q$, then the algorithm outputs $\textsf{accept}$ with probability at least $9/10$;
        \item If $|H(p)-H(q)| \geq \varepsilon$, then the algorithm outputs $\textsf{reject}$ with probability at least $9/10$.
    \end{itemize}
    Further, the algorithm runs in time $O(N \log N)$, where $N$ is the total number of samples.
\end{theorem}

\begin{algorithm}[!ht]
\caption{Entropy equivalence testing (\EET)}
\label{alg:eet}
\begin{algorithmic}[1]
\Require Sample access to $p$ and $q$, both over $[n]$; accuracy parameter $\varepsilon \in (0,1]$.
\State\label{1} Run a Hellinger closeness tester on $p,q$ with $O\!\left(\tfrac{n^{3/4}}{\varepsilon}\right)$ samples.
\If{$d_H^2(p,q) \geq \Omega(\varepsilon)$}
 {\Return \textsf{reject} }
\EndIf
\State Set $m_1 := \tilde O(n^{3/4}/\varepsilon)$. Take $m_1$ samples each from $p,q$; let $X_i, Y_i$ be the empirical counts.
\State $S \gets \bigl\{\, i \in [n] : X_i + Y_i \geq \Omega (m_1 \cdot \varepsilon / (n^{3/4} \log(n/\varepsilon))) \,\bigr\}$  \Comment{heavy elements}
\State\label{5} Set $m_2 := \tilde O(1/\varepsilon)$. Using $m_2$ samples, test whether $p(\bar S) \geq \tilde\Omega(\varepsilon)$ and $q(\bar S) \geq \tilde\Omega(\varepsilon)$.

\If{only one of the preceding inequalities hold} \Comment{$p$ and $q$ differ on $\bar{S}$.}
    \State \Return \textsf{reject}
\EndIf

\If{both preceding inequalities hold} \Comment{low-mass regime on $\bar S$}
    \State Set $m_3 := O(\log^2(n/\varepsilon)/\varepsilon^2)$. Using $m_3$ samples, estimate $|p(\bar S) - q(\bar S)|$.
    \If{$|p(\bar S) - q(\bar S)| \geq \tilde\Omega(\varepsilon)$}{ \Return \textsf{reject}}
    \EndIf
    \State Via rejection sampling on $\bar S$, run a TV closeness tester on $p_{\bar S}, q_{\bar S}$ at threshold $\tilde\Omega(\varepsilon/q(\bar{S}))$ using $\tilde O(n^{3/4}/\varepsilon)$ samples.
    \If{the conditional tester rejects}
        \State \Return \textsf{reject}
    \EndIf
\EndIf
\State Set $m_4 := \tilde O\!\bigl(n^{3/4}/\varepsilon + \log^2(n/\varepsilon)/\varepsilon^2\bigr)$ and $s := \Omega(n^{3/4} \log(n/\varepsilon)/\varepsilon)$.
\State Using $s$ samples, compute $T \gets \displaystyle\sum_{i \in S}\tfrac{(X_i - Y_i)^2 - (X_i + Y_i)}{X_i + Y_i}$. \Comment{bias check}
\If{$T \geq \Omega(\sqrt{n})$}
    \State \Return \textsf{reject}
\EndIf
\State Using $O(\log^2(n/\varepsilon)/\varepsilon^2)$ samples, test both $|p(S) - q(S)| \geq \Omega\!\bigl(\tfrac{\varepsilon}{\log(n/\varepsilon)}\bigr)$ and $\|p - q\|_2^2 \geq \Omega\!\bigl(\tfrac{\varepsilon^2}{\log^2(n/\varepsilon)}\bigr)$.
\If{either preceding inequality holds}
    \State \Return \textsf{reject}
\EndIf
\State Take $m_4$ samples; let $X_i, Y_i$ be the empirical counts on $S$. Compute $$Z \gets \displaystyle\sum_{i \in S}\frac{X_i - Y_i}{m_4}\log \tfrac{1}{X_i + Y_i}.$$ 
\If{$|Z| \geq \Omega(\varepsilon)$}{
    \Return \textsf{reject}}
\Else{
    \Return \textsf{accept}}
\EndIf
\end{algorithmic}
\end{algorithm}

The remainder of this section is dedicated to proving this theorem, which we do via a series of lemmas.  
%
%

The first lemma relates a difference in entropy to a difference in (squared) Hellinger distance and a cross entropy term with respect to a mixture:
\begin{lemma}
{\label{lemma:reduce_to_Hellinger_or_entropy}
For any two distributions $p,q$ over $[n]$, we have
\begin{align*} | H (p) - H (q) | \leq \Theta (d_H^2 (p, q)) + \left| \sum_{i=1}^n (p_i -
   q_i) \log \frac{1}{(p_i + q_i) / 2} \right| . \end{align*}}
\end{lemma}
\begin{proof}
Letting $r \coloneqq \frac{p+q}{2}$, we can write
  \begin{align*}
    &  H (p) - H (q)\\
    & =  \sum_i \left( p_i \log \frac{1}{p_i} - q_i \log \frac{1}{q_i} \right)\\
    & =  \sum_i \left( - p_i \log \frac{p_i}{(p_i + q_i) / 2} + p_i \log
    \frac{1}{(p_i + q_i) / 2} + q_i \log \frac{q_i}{(p_i + q_i) / 2} - q_i
    \log \frac{1}{(p_i + q_i) / 2} \right)\\
    & =  d_{\tmop{KL}} \left( q\,\|\, r \right) - d_{\tmop{KL}} \left( p\|
    r \right) + \sum_i \left( p_i \log \frac{1}{(p_i + q_i) / 2} - q_i
    \log \frac{1}{(p_i + q_i) / 2} \right)\\
    & =  d_{\tmop{KL}} \left( q\,\|\, r \right) - d_{\tmop{KL}} \left( p\| r \right) + \sum_i (p_i - q_i) \log \frac{1}{(p_i + q_i) /
    2} .
  \end{align*}
  We then observe that
  \begin{align*} d_{\tmop{KL}} \left( q\,\|\, r \right) \leq d_{\chi^2} \left(
     q, r \right) = \sum_i \frac{\left( q_i - \frac{p_i +
     q_i}{2} \right)^2}{\frac{p_i + q_i}{2}} = \sum_i \frac{1}{2}  \frac{(p_i
     - q_i)^2}{p_i + q_i} = \Theta (d_H^2 (p, q)), \end{align*}
  where the last equality comes from the equivalence of between Hellinger distance and Triangle distance (e.g., see \cite[Proposition 3]{DaskalakisKW18}).\footnote{Note that the induced KL terms have been studied in information theory (called Jensen-Shannon divergence; see e.g., \cite[Chapter 7]{Polyanskiy_Wu_2025} and more systematically in \cite{Topsoe00}, where they actually show that the Jensen-Shannon divergence and the Triangle distance are equivalent up to constant factors).} Similarly, we have that $d_{\tmop{KL}} \left( p\,\|\, r \right) \leq d_{\chi^2}
  \left( p, r \right) = \Theta (d_H^2 (p, q))$, and we conclude
  that
  \begin{align*}
    | H (p) - H (q) | & \leq  \left| d_{\tmop{KL}} \left( q\,\|\, r \right) \right| + \left| d_{\tmop{KL}} \left( p\,\|\, r
    \right) \right| + \left| \sum_i (p_i - q_i) \log \frac{1}{(p_i + q_i) / 2}
    \right|\\
    & \leq  \Theta (d_H^2 (p, q)) + \left| \sum_i (p_i - q_i) \log
    \frac{1}{(p_i + q_i) / 2} \right|\,,
\end{align*}
as claimed.
\end{proof}
Lemma~\ref{lemma:reduce_to_Hellinger_or_entropy} allows us to split the entropy difference 
$|H(p) - H(q)|$ into a Hellinger part and a term
$|\sum_i (p_i - q_i) \log(1/((p_i+q_i)/2))|$, so that an entropy difference $\varepsilon$ implies at least one of these two terms must be $\varepsilon/2$: thus, it suffices to test both, and reject if, and only if, at least one of them is detected to be $\Omega(\varepsilon)$. (Clearly, if $p=q$, both the Hellinger and KL-like terms are zero.)\smallskip

The Hellinger term can be
tested with $\widetilde{O}(n^{3/4}/\varepsilon)$
samples using a Hellinger closeness tester (e.g., \cite[Proposition 2.15]{DK16}): if
$d_H^2(p,q) \geq \Omega(\varepsilon)$, this tester will allow us to reject with high probability. For the rest of the subsection we therefore can assume $d_H^2(p,q) \leq O(\varepsilon)$,
and the task reduces to testing the second, ``KL-like'' term,
\begin{equation}
\Lambda(p,q) \coloneqq \left| \sum_i (p_i - q_i) \log
    \frac{1}{(p_i + q_i) / 2} \right|\,, \label{eq:cross_entropy_expression}    
\end{equation}
and specifically to distinguish $\Lambda(p,q) = 0$ from $\Lambda(p,q) \geq \frac{\varepsilon}{2}$.

To do so, the idea is to handle separately the set $S^\ast$ of ``heavy'' elements, on which either $p$ or $q$ puts relatively high probability (specifically, at least some threshold $\tau \coloneqq\tilde{\Omega}(\varepsilon/n^{3/4})$), and the remaining set $\bar{S^\ast} = [n]\setminus S^\ast$ of ``light'' elements (which have low probability under both $p$ and $q$). 
Of course, we do not know this set $S^\ast$, as both $p,q$ are unknown; however, by taking $\tilde{O}(1/\tau)=\tilde{O}(n^{3/4}/\varepsilon)$ samples (which is inconsequential, as we will pay this sample complexity for the Hellinger test anyway), we can obtain a set $S$ which is a good proxy for $S^\ast$. To make this formal, define the following two sets:
\begin{align*} S_1 &\assign \left\{ i \in [n] \mid p_i + q_i \geq C_1 \cdot
   \frac{\varepsilon}{n^{3 / 4} \log (n / \varepsilon)} \right\} . \\
   S_2 &\assign \left\{ i \in [n] \mid p_i + q_i \geq C_2 \cdot
   \frac{\varepsilon}{n^{3 / 4} \log (n / \varepsilon)} \right\} , \end{align*}
where $C_2 \geq C_1 > 0$ are two absolute constants. Our focus is to identify a set $S\subseteq [n]$ such that $S_2 \subseteq S \subseteq S_1$.
\begin{lemma}\label{lemma:filtering_small_elements}
  Using $O \left( \frac{n^{3 / 4} \log(n/\varepsilon) \log n}{\varepsilon} \right)$ samples from $p$ and $q$, one can efficiently return a set $S$ satisfying $S_2
  \subseteq S \subseteq S_1$ with probability at least $\frac{99}{100}$.
\end{lemma}
\begin{proof}
By \cref{fact:distinguish_small_vs_big}, choosing $2C_1 \leq C_2$, as we have sample access to the uniform mixture $\frac{p+q}{2}$, we can distinguish $\frac{p_i+q_i}{2} \geq C_2 \cdot \frac{\varepsilon}{2n^{3 / 4} \log (n /
\varepsilon)}$ and $\frac{p_i+q_i}{2} \leq C_1 \cdot \frac{\varepsilon}{2n^{3 / 4}
\log (n / \varepsilon)}$ with probability $1 - \frac{1}{100 n}$ costing $O \left(
\frac{n^{3 / 4} \log (n / \varepsilon) \cdot \log n}{\varepsilon} \right)$
samples from $p$ and $q$. Therefore, running the test over the support $[n]$ (using the same set of samples to run this test for all $i$'s) and adding each $i\in[n]$ to the set $S$ if we detect $p_i + q_i \geqslant C_2 \cdot \frac{\varepsilon}{n^{3 / 4} \log (n / \varepsilon)}$, by a union bound, the algorithm will be simultaneously correct on all $i\in[n]$  with probability at least $\frac{99}{100}$.
\end{proof}
To summarize, once we have identified such a set $S$, given that
\[
\underbrace{\Lambda(p,q)}_{0\text{ or } \geq \frac{\varepsilon}{2}}
\leq \underbrace{\left| \sum_{i \in \overline{S}} (p_i - q_i) \cdot \log \frac{1}{p_i + q_i} \right|}_{\eqcolon \Lambda_{\overline{S}}(p,q)} + \underbrace{\left| \sum_{i \in S} (p_i - q_i) \cdot \log \frac{1}{p_i + q_i} \right|}_{\eqcolon \Lambda_S(p,q)}
\]
it suffices to separately test (1)  $\Lambda_{\overline{S}}(p,q)=0$ vs.~$\Lambda_{\overline{S}}(p,q) \geq \varepsilon/4$, and (2) $\Lambda_S(p,q) = 0$ vs.~$\Lambda_S(p,q) \geq \varepsilon/4$; we do so below, assuming that our set $S$ satisfies $S_2 \subseteq S$.\footnote{For simplicity of exposition, we henceforth ignore the factor $4$, and consider $\varepsilon$ instead of $\varepsilon/4$; a simple rescaling argument gives the original guarantees.}

\subsection*{Test 1: $\Lambda_{\overline{S}}(p,q)=0$ vs.~$\Lambda_{\overline{S}}(p,q) \geq \Omega(\varepsilon)$}

Inspired by the ideas of \citet{DK16}, we can then check if the distributions are far \emph{in TV distance} in the low-probability regime (i.e., elements in $\overline{S}$), which will suffice by a simple reduction from entropy difference to
TV (using the lower bound $\min_i p_i, \min_i q_i \geq \tilde{\Omega}(\varepsilon/n)$). In more detail, we have
\begin{align*} 
\bigg| \sum_{i \in \overline{S}} (p_i - q_i) \cdot \log \frac{1}{p_i + q_i} \bigg| \leq
\sum_{i \in \overline{S}} | p_i - q_i | \cdot \log \frac{1}{p_i + q_i} \leq
O(\log (n / \varepsilon)) \cdot \sum_{i \in \overline{S}} | p_i - q_i | . 
\end{align*}

While this is a task that would in general require a large number of samples, here we can leverage the additional fact that $\overline{S}$ does not contain any heavy elements.

\begin{lemma}
  Let $T\subseteq [n]$ be a set such that $\max_{i\in T} (p_i + q_i) \leq C_2 \cdot
  \frac{\varepsilon}{n^{3 / 4} \log (n / \varepsilon)}$. Then, testing $p_T =
  q_T$ versus $\sum_{i \in T} | p_i - q_i | \geq \frac{\varepsilon}{\log (n /
  \varepsilon)}$ (with probability of success $9/10$) can be done with $O \!\left( {n^{3 / 4} \cdot \log(n/\varepsilon)}/{\varepsilon}
  \right)$ samples from $p,q$.
\end{lemma}
\begin{proof}
By~\cite{DK16}, we know that $\tilde{O} \left(
\frac{n^{3 / 4}}{\varepsilon} \right)$ samples suffice to distinguish $p_T=q_T$ from $d_{\tmop{TV}}(p_{T}, q_{T}) \geq \varepsilon$ with high constant probability. For the sake of completeness, we briefly explain how to adapt their argument to our setting. First, it takes $O(\frac{\log(n/\varepsilon)}{\varepsilon})$ samples to check whether the mass of $p$ on $T$ is larger than $\Omega(\varepsilon/\log(n/\varepsilon))$, and similarly for $q$: if both have mass smaller that this, we accept, as $\sum_{i \in \overline{S}} | p_i - q_i | \leq p(T)+q(T) < \frac{\varepsilon}{\log (n /
  \varepsilon)}$. If only one of the two puts large enough mass on $T$, we reject, as this shows that $\sum_{i \in T} | p_i - q_i | \neq 0$. Otherwise, both $p,q$ put enough mass on $T$, and we continue.\smallskip

With $O
\left( \frac{\log^2 (n / \varepsilon)}{\varepsilon^2} \right)$ samples, we then confirm that
their probability masses on $T$ approximately match, i.e., that $| p (T) - q (T) | \leq O \left(
\frac{\varepsilon}{\log (n / \varepsilon)} \right)$. %
If they do not, we can output reject. Otherwise, we have
\begin{align*}
  \sum_{i \in T} | p_i - q_i |  &= \sum_{i \in T} | p(i | i \in T) p
  (T) - q (i | i \in T) q (T) |\\
  &\leq \sum_{i \in T} ( | p(i | i \in T) p (T) - p(i | i \in T) q
  (T) | + | p(i | i \in T) q (T) - q(i | i \in T) q (T) | )\\
  &= ~ | p (T) - q (T) | + q (T) \cdot \| p(\cdot | T) - q(\cdot | T) \|_1
  .
\end{align*}
Thus, if $\| p_{T} - q_{T} \| \geq \Omega \left(
\frac{\varepsilon}{\log (n / \varepsilon)} \right)$ and $| p (T) - q (T) |
\leq O \left( \frac{\varepsilon}{\log (n / \varepsilon)} \right)$ (for suitable choice of constants in the $O(\cdot)$, $\Omega(\cdot)$), then
we have that
\begin{align*} 
d_{\rm TV}(p_T,q_T) = \frac{1}{2} \| p(\cdot|T)-q(\cdot|T) \|_1 \geq \Omega\left(\frac{\varepsilon}{q
   (T) \cdot \log (n / \varepsilon)} \right).
\end{align*}
Recalling that $q (T),p (T) \geq \frac{\varepsilon}{\log (n / \varepsilon)}$ and $q (T) = \Theta \left( p (T) \right)$, we can sample from $p_T,q_T$ via rejection sampling at the cost of an expected $O(1/q(T))$, $O(1/p(T))$ factor in the sample complexity. Thus, we can use an off-the-shelf TV equivalence tester (e.g., from \citep[Theorem 1]{CDVV14}) to test $d_{\rm TV}(p(\cdot|T),q(\cdot|T)) \geq \Omega\left(\frac{\varepsilon}{q
   (T) \cdot \log (n / \varepsilon)} \right)$ vs.~$p(\cdot|T)=q(\cdot|T)$ at the cost of 
\begin{align*} \max \left\{ \frac{n^{2 / 3} q^{4 / 3} (T) \log^{4 / 3} (n /
   \varepsilon)}{\varepsilon^{4 / 3}}, \frac{n^{1 / 2} q^2 (T) \log^2 (n /
   \varepsilon)}{\varepsilon^2} \right\}
\end{align*}
samples from $p(\cdot|T),q(\cdot|T)$; which, by the above discussion, requires an expected
\begin{align*} \max \left\{ \frac{n^{2 / 3} q^{1 / 3} (T) \log^{4 / 3} (n /
   \varepsilon)}{\varepsilon^{4 / 3}}, \frac{n^{1 / 2} q (T) \log^2 (n /
   \varepsilon)}{\varepsilon^2} \right\} 
\end{align*}
samples from $p,q$ (which can be made high-probability via Markov's inequality, paying a constant factor in the sample complexity and a small additional probability of failure).
Since $q (T) \leq \frac{\varepsilon}{n^{3 / 4} \log(n/\varepsilon)} \cdot n =
\frac{\varepsilon \cdot n^{1 / 4}}{\log(n/\varepsilon)}$, we finally conclude that taking $O \left( \frac{n^{3
/ 4} \cdot \log(n/\varepsilon)}{\varepsilon} \right)$ samples suffices.
\end{proof}
\subsection*{Test 2: $\Lambda_S(p,q) = 0$ vs.~$\Lambda_S(p,q) \geq \Omega(\varepsilon)$}

We now look at the second~--~and more involved~--~case, namely, detecting when the entropy difference is far in the high mass regime.  That is, we are testing $\Lambda_S(p,q) = 0$ vs.~$\Lambda_S(p,q) \geq \varepsilon$.

We will define a suitable test statistic, $Z$, and analyze its expectation (with the goal of bounding its bias and relate it to $\Lambda_S(p,q)$), before bounding its variance to conclude with an application of Chebyshev's inequality. At a high level, our strategy is as
follows: either the empirical entropy difference $Z$ has \emph{low bias}, which means we
can simply rely on its output for the test, or it has \emph{high bias}, in which case this can be detected by performing additional checks to detect this bias.\smallskip

Let $m = \tilde{O} \left( \frac{n^{3 /
4}}{\varepsilon} + \frac{\log (n /
\varepsilon)}{\varepsilon^2} \right)$ be the number of samples taken from $p,q$, and assume the Poissonized setting: that is, instead of exactly $m$ samples, we take a random $\tmop{Poi}(m)$ number of samples.\footnote{This standard technique is valid as it only affects the sample complexity by constant terms and an arbitrarily small error probability adjustment, since a Poisson random variable tightly concentrates around its mean.  Its advantage is to provide additional independence between the empirical counts of each element, which is useful in the analysis.} Recall that by construction of the set $S\subseteq S_1$, we can assume $p_i+ q_i \geq \tilde{\Omega} \left(
\frac{\varepsilon}{n^{3 / 4}} \right)$ for all $i\in S$. Focusing on $S$, we run two further checks, each at cost
$O(\log^2 m/\varepsilon^2)$:
\begin{enumerate}
\item $|p(S) - q(S)| \leq O(\varepsilon/\log m)$, which implies
$|\sum_{i \in S}(p_i - q_i)| \log m \leq O(\varepsilon)$.
\item $\|p - q\|_2^2 \leq O(\varepsilon^2/\log^2 m)$, which is checked via testing at distance
$\varepsilon/\log m$ using an $\ell_2$-closeness tester (e.g., via learning $p,q$ in $\ell_2$ distance, see for instance~\cite{Canonne2020shortnotelearningdiscrete}).
\end{enumerate}
If either check rejects, we reject; otherwise we proceed assuming both guarantees
hold.

Letting $X_i$ (resp.~$Y_i$) be the number of occurrences of element $i\in S$ among the samples from $p$ (resp.~from $q$), we define our entropy difference estimator to be
\begin{align*}
Z := \sum_{i \in S} \frac{X_i - Y_i}{m}\,\log\frac{1}{X_i + Y_i},
\end{align*}
and we write $Z_i := \tfrac{X_i - Y_i}{m}\log\tfrac{1}{X_i + Y_i}$ for the
$i$-th summand and let $Z_i=0$ when $X_i+Y_i=0$. By the above Poissonization procedure, we have $X_i \sim \tmop{Poi}(m \cdot p_i)$ and $Y_i \sim \tmop{Poi}(m \cdot q_i)$ for all $i$, and the $(X_i)_i$ (resp.  $(Y_i)_i$) are independent.
We first bound the bias of this estimator $Z$:
\begin{lemma}
\label{lemma:upper_bound_of_expected_error_Z}
For any set $S \subseteq [n]$,
\begin{align*}
\left|\sum_{i \in S} (p_i - q_i)\log\tfrac{1}{m(p_i + q_i)}
      - \mathbb{E}[Z]\right|
\leq \sum_{i \in S}\frac{|p_i - q_i|}{m(p_i + q_i)}.
\end{align*}
\end{lemma}
\begin{proof}
Let $J_i := X_i + Y_i \sim \mathrm{Poi}(\lambda_i)$ with
$\lambda_i := m(p_i + q_i)$. Conditional on $J_i = j$, the count $X_i$ is
$\mathrm{Bin}(j, p_i/(p_i+q_i))$-distributed, so
$\mathbb{E}[X_i - Y_i \mid J_i] = J_i \cdot (p_i - q_i)/(p_i + q_i)$.
By the law of total expectation, pulling the factor $\log(1/J_i)$ outside
the inner expectation,
\begin{align*}
\mathbb{E}[Z_i]
= \tfrac{1}{m}\,\mathbb{E}\left[(X_i - Y_i)\log\tfrac{1}{J_i}\right]
= \tfrac{1}{m}\,\mathbb{E}\left[\log\tfrac{1}{J_i}\cdot\mathbb{E}[X_i - Y_i \mid J_i]\right]
= \tfrac{p_i - q_i}{m(p_i+q_i)}\,\mathbb{E}[J_i \log(1/J_i)].
\end{align*}
We then apply the Poisson identity (Fact \ref{fact:factorial_moment_functional_shifting}) $\mathbb{E}[J_i f(J_i)] = \lambda_i\,\mathbb{E}[f(J_i+1)]$, taking $f(j) = \log(1/j)$:
\begin{align*}
\mathbb{E}[Z_i]
= \tfrac{p_i - q_i}{m(p_i+q_i)}\cdot \lambda_i\,\mathbb{E}[\log(1/(J_i+1))]
= -(p_i - q_i)\,\mathbb{E}[\log(J_i + 1)].
\end{align*}
Subtracting the target $(p_i - q_i)\log(1/\lambda_i)$,
\begin{align*}
(p_i - q_i)\log\tfrac{1}{\lambda_i} - \mathbb{E}[Z_i]
= (p_i - q_i)\,\mathbb{E}\left[\log\tfrac{J_i + 1}{\lambda_i}\right].
\end{align*}
We proceed to bound $\mathbb{E}[\log((J_i+1)/\lambda_i)]$ from above and below. By
Jensen's inequality applied to the concave function $\log$,
\begin{align*}
\mathbb{E}\!\left[\log\tfrac{J_i+1}{\lambda_i}\right]
\leq \log\!\left(1 + \tfrac{1}{\lambda_i}\right)
\leq \tfrac{1}{\lambda_i}.
\end{align*}
For the lower bound, the inequality $\log x \geq 1 - 1/x$ applied to
$x = (J_i+1)/\lambda_i$, together with
$\mathbb{E}[1/(J_i+1)] = (1 - e^{-\lambda_i})/\lambda_i$ for
$J_i \sim \mathrm{Poi}(\lambda_i)$, gives
\begin{align*}
\mathbb{E}\!\left[\log\tfrac{J_i+1}{\lambda_i}\right]
\geq 1 - \lambda_i\,\mathbb{E}\!\left[\tfrac{1}{J_i+1}\right]
= e^{-\lambda_i}
\geq 0.
\end{align*}
Combining these findings gives $|\mathbb{E}[\log((J_i+1)/\lambda_i)]| \leq 1/\lambda_i$
(using $e^{-\lambda_i} \leq 1/\lambda_i$ for all $\lambda_i > 0$), so
\begin{align*}
\bigl|(p_i - q_i)\log\tfrac{1}{\lambda_i} - \mathbb{E}[Z_i]\bigr|
\leq \tfrac{|p_i - q_i|}{\lambda_i}.
\end{align*}
Summing over $i \in S$ yields the claim.
\end{proof}

The bias of our estimator is determined by the RHS of \cref{lemma:upper_bound_of_expected_error_Z}; note that if $p
= q$, the RHS simplifies to $0$. As a result, a natural idea would be to use this RHS as a proxy, and to check if the RHS is noticeably large using another test: if it is larger than $\geq \Omega
(\varepsilon)$, we can immediately reject, as this is clear evidence that $p\neq q$. We do so in the following lemma:
\begin{lemma}
\label{lemma:check_if_Z_has_small_bias}
Suppose that $d_H^2(p,q) \leq O(\varepsilon)$. Then, with $s=O(n^{3/4}\log(n/\varepsilon)/\varepsilon)$
samples from $p,q$, one can distinguish between the following two cases with high constant probability:
\begin{enumerate}
    \item  $p = q$ (which implies $\sum_{i \in S} \frac{| p_i - q_i |}{m (p_i + q_i)} =
  0$).\label{case:a}
  \item $\sum_{i \in S} \frac{| p_i - q_i |}{(p_i + q_i)} \geq \Omega (m
  \varepsilon)$.\label{case:b}
\end{enumerate}
\end{lemma}

\begin{proof}
We use the test statistic
\begin{align*}
T \coloneqq \sum_{i \in S}\frac{(X_i - Y_i)^2 - (X_i + Y_i)}{X_i + Y_i},
\end{align*}
where $X_i \sim \tmop{Poi}(s \cdot p_i)$ and $Y_i \sim \tmop{Poi}(s \cdot q_i)$.
We show:
\begin{itemize}
\item In the first case, $\mathbb{E}[T] =0$, while in the second case $\mathbb{E}[T] \geq \Omega(\sqrt{n})$.
\item $\mathrm{Var}[T] \leq O(n)$ in both cases.
\end{itemize}
Chebyshev's inequality then will enable us to separate the two cases by thresholding $T$ at a value scaling as $\Theta(\sqrt{n})$.

\paragraph*{Analyzing $\mathbb{E}[T]$.} The first part is similar to the analysis from~\cite[Lemma~4]{CDVV14}, which relies on the same statistic $T$ (for a different purpose):  write $\delta_i := (p_i - q_i)/(p_i + q_i)$ and $J_i := X_i + Y_i$.
Conditional on $J_i = j$, we have $Y_i = j - X_i$ so $X_i - Y_i = 2X_i - j$, and
$X_i \sim \mathrm{Bin}(j, p_i/(p_i+q_i))$. Hence,
\begin{align*}
\mathbb{E}[X_i - Y_i \mid J_i = j] = j\delta_i,
\end{align*}
and the binomial variance $\mathrm{Var}[X_i \mid J_i = j] =
j\cdot\tfrac{p_i q_i}{(p_i+q_i)^2}$ gives
\begin{align*}
\mathrm{Var}[X_i - Y_i \mid J_i = j] = 4j\cdot\tfrac{p_i q_i}{(p_i+q_i)^2} = j(1 - \delta_i^2).
\end{align*}
Combining, $\mathbb{E}[(X_i - Y_i)^2 \mid J_i = j] = \mathrm{Var}[X_i - Y_i \mid J_i = j] + \mathbb{E}[X_i - Y_i \mid J_i = j]^2 = j(1 - \delta_i^2) + j^2\delta_i^2$. The $i$-th summand of $T$ is $\frac{(X_i - Y_i)^2 - J_i}{J_i}$, which we set to $0$ when $J_i = 0$ (the $0/0$ case). For $J_i = j \geq 1$,
\begin{align*}
\mathbb{E}\!\left[\frac{(X_i - Y_i)^2 - J_i}{J_i} \,\Big|\, J_i = j\right]
= \frac{(j^2 - j)\delta_i^2}{j} = (j - 1)\delta_i^2.
\end{align*}
Taking the outer expectation, with the indicator $\mathbf{1}\{J_i \geq 1\}$ enforcing the $0/0$ convention,
\begin{align*}
\mathbb{E}[T]
= \sum_{i \in S} \delta_i^2 \cdot \mathbb{E}[(J_i - 1)\mathbf{1}\{J_i \geq 1\}]
= \sum_{i \in S} \delta_i^2 \bigl(\lambda_i - 1 + e^{-\lambda_i}\bigr),
\end{align*}
where $\mathbb{E}[(J_i - 1)\mathbf{1}\{J_i \geq 1\}] = \mathbb{E}[J_i] - \Pr[J_i \geq 1] = \lambda_i - (1 - e^{-\lambda_i})$.

In particular, in the first case (when $p=q$), $\delta_i=0$ for every $i\in S$, and so $\mathbb{E}[T]=0$. We therefore now focus on the second case, for which we depart from~\cite{CDVV14}. By the construction of $S$ (\cref{lemma:filtering_small_elements}), $p_i + q_i \geq \Omega(\frac{\varepsilon}{n^{3/4} \log(n/\varepsilon)})$ for every $i \in S$, so choosing $s = \Omega(n^{3/4}\cdot\log(n/\varepsilon)/\varepsilon)$ with a sufficiently large constant gives $\lambda_i = s(p_i+q_i) \geq 2$ for all $i \in S$.  The bound $\lambda - 1 + e^{-\lambda} \geq \lambda - 1 \geq \lambda/2$ for $\lambda \geq 2$ then yields
\begin{align*}
\mathbb{E}[T]
\geq \tfrac{1}{2}\sum_{i \in S}\lambda_i\delta_i^2
= \tfrac{s}{2}\sum_{i \in S}\frac{(p_i - q_i)^2}{p_i + q_i}.
\end{align*}
By Cauchy--Schwarz,
\begin{align*}
\sum_{i \in S}\frac{(p_i - q_i)^2}{p_i + q_i}
\geq \frac{1}{n}\left(\sum_{i \in S}\frac{|p_i - q_i|}{\sqrt{p_i + q_i}}\right)^{\!2}.
\end{align*}
Since $m(p_i+q_i) \geq 1$ (in view of $i \in S$ and $m=\tilde{\Omega}(\frac{n^{3/4}}{\varepsilon})$), we have 
$\frac{1}{\sqrt{p_i+q_i}} \geq \frac{1}{\sqrt{m}(p_i+q_i)}$.
Applying this termwise and using the definition of the second case,
\begin{align*}
\sum_{i \in S}\frac{|p_i - q_i|}{\sqrt{p_i + q_i}}
\geq \frac{1}{\sqrt{m}}\sum_{i \in S}\frac{|p_i - q_i|}{p_i+q_i}
\geq \Omega(\sqrt{m}\varepsilon).
\end{align*}
Combining,
\begin{align*}
\mathbb{E}[T]
\geq \frac{s}{2n}\cdot\Omega(m\varepsilon^2)
= \Omega\!\left(\frac{s m\varepsilon^2}{n}\right),
\end{align*}
which is $\Omega(\sqrt{n}+\frac{\log^2 (n/\varepsilon)}{n^{1/4} \varepsilon})\geq\Omega(\sqrt{n})$ given the settings of $s = \tilde{\Theta}(n^{3/4}/\varepsilon)$ and $m = \tilde{\Theta}(n^{3/4}/\varepsilon + \frac{\log^2(n/\varepsilon)}{\varepsilon^2})$.
\paragraph*{Upper bound on $\mathrm{Var}[T]$.}
The variance bound from \cite[Lemma 5]{CDVV14} gives
\begin{align*}
\mathrm{Var}[T]
\leq 2\min\{n, s\}
+ 5s\sum_{i \in S}\frac{(p_i - q_i)^2}{p_i + q_i}.
\end{align*}
Under the lemma's assumption on the Hellinger distance,
$\sum_{i \in S}(p_i - q_i)^2/(p_i + q_i) = \Theta(d_H^2(p, q)) \leq O(\varepsilon)$,
so the second term is $\widetilde{O}(s \cdot \varepsilon)$, and
$\mathrm{Var}[T] \leq O(n + s \cdot \varepsilon)$, where $s = \tilde{O} \left( n^{3 / 4}/\varepsilon \right)$.  This scaling of $s$ leads to the simplified bound $\mathrm{Var}[T] \leq O(n)$. 

\paragraph*{Concluding via Chebyshev.}
In the first case, $\mathbb{E}[T] = 0$ and $\mathrm{Var}[T] \leq O(n)$, so
$|T| \leq O(\sqrt{n})$ with probability $\geq 9/10$. In the second case,
$\mathbb{E}[T] \geq C\sqrt{n}$ and
$\mathrm{Var}[T] \leq O(n) \leq \mathbb{E}[T]^2/C'$, so
$T \geq C''\sqrt{n}$ with probability $\geq 9/10$. Thresholding
at $\Theta(\sqrt{n})$ thus distinguishes the two cases.
\end{proof}

Having used the RHS of~\cref{lemma:upper_bound_of_expected_error_Z} as a proxy for the bias, and checked that it is small (and rejected otherwise), we are not yet done: we could still have both small bias for $Z$ (so that $\mathbb{E}[Z] \approx \sum_{i \in S} (p_i - q_i)\log\tfrac{1}{m(p_i + q_i)}$) \emph{and} large $\Lambda_S(p,q) = \left|\sum_{i \in S} (p_i - q_i)\log\tfrac{1}{m(p_i + q_i)}\right|$. However, now, since we know that $|\mathbb{E}[Z]| \approx \Lambda_S(p,q)$, to conclude it suffices to run one last check: namely, that $Z$ is small enough. 

Indeed, we know now that $Z$ has small bias: if we can also show, via a variance argument, that $Z$ concentrates around its expectation, then we will have $|Z|\approx_{\rm w.h.p.} |\mathbb{E}[Z]| \approx \Lambda_S(p,q)$, and checking that $|Z|$ is small will ensure (with high probability) that we can guarantee that $\Lambda_S(p,q)$ is small. Thus, our last missing piece is an upper bound on $\tmop{Var} [Z]$, which will let us show that $Z$ concentrates around its expectation.

\begin{lemma}\label{lem:varbound}
  The variance of the entropy difference estimator $Z$ satisfies
  \begin{align*} \tmop{Var} [Z] \leq O (\log^2 m)  \| p - q \|_2^2 + \frac{\log^2 m}{m} \leq O (\varepsilon^2)
     . \end{align*}
\end{lemma}

\begin{proof}
Under Poissonization, the counts $\{X_i\}_{i \in S}$ and $\{Y_i\}_{i \in S}$ are
mutually independent across coordinates.  As a result, the $Z_i$'s are also independent, and
\begin{align*}
\mathrm{Var}[Z] = \sum_{i \in S}\mathrm{Var}[Z_i]
               \leq \sum_{i \in S}\mathbb{E}[Z_i^2].
\end{align*}

Fix any $i\in S$, letting $J_i := X_i + Y_i \sim \mathrm{Poi}(m(p_i + q_i))$, and defining the event
$\mathcal{E}_i := \{J_i \leq 3m\}$. Since $m(p_i+q_i) \leq m$, the Chernoff bound for Poisson random variables gives $\Pr[\mathcal{E}_i^c] \leq e^{-\Omega(m)}$.  Then:
\begin{itemize}
    \item On $\mathcal{E}_i$: $|\log(1/J_i)| \leq \log(3m)$ for $J_i \ne 0$, and noting from before that $Z_i=0$ when $J_i=0$, we have
$|Z_i|\,\mathbf{1}_{\mathcal{E}_i} \leq \tfrac{\log(3m)}{m}|X_i - Y_i|$.
Squaring and taking expectations,
\begin{align*}
\mathbb{E}[Z_i^2\,\mathbf{1}_{\mathcal{E}_i}]
\leq \tfrac{\log^2(3m)}{m^2}\,\mathbb{E}[(X_i - Y_i)^2]
= \tfrac{\log^2(3m)}{m^2}\bigl(m(p_i+q_i) + m^2(p_i-q_i)^2\bigr),
\end{align*}
using $\mathrm{Var}[X_i-Y_i] = m(p_i+q_i)$ and
$\mathbb{E}[X_i-Y_i] = m(p_i-q_i)$ (moments of Poisson r.v.'s).

\item On $\mathcal{E}_i^c$: from $|X_i - Y_i| \leq J_i$ and
$|\log(1/J_i)| = \log J_i \leq J_i$ (valid since $J_i > 3m \geq 1$),
we get $|Z_i|\,\mathbf{1}_{\mathcal{E}_i^c} \leq J_i^2/m$. Combining
Cauchy--Schwarz with the moment bound $\mathbb{E}[J_i^8] = O(m^8)$ (since Poisson
moments are polynomials in the mean, by~\cref{fact:factorial_moment_functional_shifting}) and the tail bound
$\Pr[\mathcal{E}_i^c] \leq e^{-\Omega(m)}$,
\begin{align*}
\mathbb{E}[Z_i^2\,\mathbf{1}_{\mathcal{E}_i^c}]
\leq \tfrac{1}{m^2}\,\mathbb{E}[J_i^4\,\mathbf{1}_{\mathcal{E}_i^c}]
\leq \tfrac{1}{m^2}\sqrt{\mathbb{E}[J_i^8]\,\Pr[\mathcal{E}_i^c]}
\leq e^{-\Omega(m)}.
\end{align*}

\end{itemize}

Summing over $i \in S$, using $\log^2(3m) \leq 4\log^2 m$ (for
$m \geq 3$) and $\sum_i(p_i+q_i) \leq 2$,
\begin{align*}
\mathrm{Var}[Z]
\leq O(\log^2 m)\,\|p-q\|_2^2
+ \tfrac{O(\log^2 m)}{m}
+ n \cdot e^{-\Omega(m)}.
\end{align*}
The last term is absorbed into $O(\tfrac{\log^2 m}{m})$ since
$m \geq \Omega(\log n)$ makes $n \cdot e^{-\Omega(m)} \leq 1/m$.
\end{proof}

\noindent
Combining Lemma~\ref{lemma:upper_bound_of_expected_error_Z} and
Lemma~\ref{lemma:check_if_Z_has_small_bias}, and conditioning on all tests passing, the expectation of $Z$ has small bias:
\begin{align*}
\left|\mathbb{E}[Z] - \sum_i (p_i - q_i)\log\tfrac{1}{(p_i + q_i)}\right|
\leq O(\varepsilon).
\end{align*}
By Lemma~\ref{lem:varbound}, $\mathrm{Var}[Z] \leq O(\varepsilon^2)$,
so Chebyshev gives $|Z - \mathbb{E}[Z]| \leq O(\varepsilon)$ with
probability at least $9/10$. Combining via the triangle inequality,
\begin{align*}
\Pr\left[\,\left|Z - \sum_i (p_i - q_i)\log\tfrac{1}{(p_i + q_i)}\right|
\leq O(\varepsilon)\,\right]
\geq \tfrac{9}{10}.
\end{align*}
\paragraph*{To conclude this subcase:} We choose a threshold for $|Z|$ at a suitable constant multiple of $\varepsilon$ that distinguishes
the two cases. If $\Lambda(p,q)=\left|\sum_{i\in S} (p_i - q_i)\log\tfrac{1}{(p_i+q_i)}\right| = 0$, then with probability at least $9/10$ the test on $T$ from Lemma~\ref{lemma:check_if_Z_has_small_bias} will not reject and $|Z| \leq O(\varepsilon)$, and so this overall test accepts. If $d_H^2(p,q) \leq O(\varepsilon)$ and $\bigl|\sum_{i\in S} (p_i - q_i)\log\tfrac{1}{(p_i+q_i)}\bigr| \geq \Omega(\varepsilon)$, then with probability at least $9/10$ either the test on $T$ from Lemma~\ref{lemma:check_if_Z_has_small_bias} rejects, or it does not and then $|Z| \geq \Omega(\varepsilon)$, and this overall test rejects.

\subsection*{Wrapping up}

Putting the above ingredients together:
\begin{itemize}
    \item If $p=q$, then with high constant probability, (1) the Hellinger test on $d_H(p,q)$ accepts, and (2) (a) the set $S$ is correctly identified, (b) the test on $\Lambda_{\overline{S}}(p,q)$ accepts, (c) (i) the test on $T$ (checking for low bias for $Z$) accepts, and (ii) the test on $Z$ (test on $\Lambda_{S}(p,q)$ assuming low bias for $Z$) accepts. Overall, the algorithm accepts with probability at least $9/10$ by a union bound.
    \item If $|H(p)-H(q)| \geq \varepsilon$, then with high constant probability (1) either the Hellinger test on $d_H(p,q)$ rejects; or it accepts because $d_H(p,q)$ is small, in which case $\Lambda(p,q)$ must be large; then, with high constant probability (2) (a) the set $S$ is correctly identified, (b) the test on $\Lambda_{\overline{S}}(p,q)$ rejects, or it accepts because $\Lambda_{\overline{S}}(p,q)$ is small, which means $\Lambda_{S}(p,q)$ must be large: in which case (c) (i) the test on $T$ (checking for low bias for $Z$ assuming $d_H(p,q)$ is small) rejects, or it accepts because the bias is small, in which case (ii) the test on $Z$ (test on $\Lambda_{S}(p,q)$ assuming low bias for $Z$) rejects. Overall, the algorithm rejects with probability at least $9/10$ by a union bound.
\end{itemize}
Overall, each subtest in \cref{alg:eet} uses no more than $O\left(\frac{n^{3/4} \log(n/\varepsilon)\log n}{\varepsilon}\right)$ or $O\left(\frac{\log^2(n/\varepsilon)}{\varepsilon^2}\right)$ samples, and hence the total sample complexity of our algorithm is $$O \left(\frac{n^{3/4} \log(n/\varepsilon)\log n}{\varepsilon} + \frac{\log^2(n/\varepsilon)}{\varepsilon^2}\right).$$
All our test statistics can be computed in $O(n)$ time once all $X_i, Y_i$ are collected.\footnote{The statistics for testing Hellinger square is the same as $T$ (albeit with different threshold; see \cite[Theorem 5]{DaskalakisKW18}). $T$ can be computed in $O(n)$ time once all $X_i$, $Y_i$ are known. Computing $S \gets \bigl\{\, i \in [n] : X_i + Y_i \geq \Omega (m_1 \cdot \varepsilon / (n^{3/4} \log(n/\varepsilon))) \,\bigr\}$ can be done in $O(n)$ time and calculating the empirical of $S$ can be implemented in $O(m_2)$ time. As for the conditional TV tester, the rejection sampling will cost as many samples as we take in total, and for TV closeness testing we use the same statistic $T$ (see, e.g., \cite[Theorem 4]{DaskalakisKW18}). For $\ell_2$ tester, the statistic is also computable in $O(n)$ time \cite[Section 3]{CDVV14}. Assuming that we can compute $\log(x)$ in constant time, one can compute the statistic $Z$ in $O(n)$ time.}
To compute $X_i, Y_i$, one can first sort over all samples which takes $O(s \log s)$ time, where $s$ is the number of samples in one of the subtests, and then sum them up. Thus, the time complexity of our algorithm is $O(N \log N)$, where $N$ is the total number of samples.
This concludes the proof of~\cref{theo:eet}.

\section{Bayesian Network Equivalence Testing}
In this section, we show how to leverage our Entropy Equivalence Testing (\EET) algorithm (\cref{alg:bayesnet}) to obtain an \emph{equivalence testing} (closeness testing) algorithm for bounded-degree Bayesian networks (Bayes nets). Specifically, the problem we consider is the following:
\begin{framed}
\noindent\textbf{Bayes-Net Equivalence Testing}:
Given sample access to distributions $p,q$ over $\{0,1\}^n$, each Markov with respect to a DAG of in-degree at most $d$ (the DAG is unknown, and possibly distinct between $p$ and $q$), and parameters $\varepsilon\in(0,1/2]$ and $\delta\in(0,1]$, distinguish, with probability at least $1-\delta$, between the following two cases: (1)~$p=q$, and (2)~$d_{\rm TV}(p,q) \ge \varepsilon$.
\end{framed}

\begin{definition}\label{definition:projection_to_DAG}
  A \emph{projection of a Bayes net} $p$ on $\{ 0, 1 \}^n$ onto a DAG $G$ is denoted
  by $p_G$, and is defined by its probability mass function (PMF) as follows:
  \[ p_G (X_1, \ldots, X_n) = \prod_{i = 1}^n p (X_i \mid \Pi_i^G), \]
  where $\Pi_i^G$ is the set of parents of $X_i$ in $G$. 
  Moreover, with a slight abuse of notation, we write $p_{X_i, \Pi_i}$ or $p_{X_i, \Pi_i}(x_i, \pi_i)$ for the marginal distribution of $p$ on the subset $\{ X_i, \Pi_i \}.$
\end{definition}

The baseline here is the testing-by-learning approach, whereby one learns the two unknown in-degree $d$ Bayes nets to TV distance $\frac{\varepsilon}{3}$ and checks the distance between the two learned hypotheses. It is known that learning an in-degree $d$ Bayes net in total variation distance can be done in $\tilde{O}(2^d n / \varepsilon^2)$ samples (see~\cite[Appendix A]{CDKS20}): to the best of our knowledge, the time complexity of any learning algorithm achieving this bound is exponential, namely, $2^{\tilde{\Omega}(n)}$. In contrast, we obtain a testing algorithm with drastically better running time (for all $d=o(n)$), and significantly better sample complexity whenever $d \gg \log(n/\varepsilon^2)$.

\begin{theorem}[\cref{theo:bayesnet:closeness:intro}, restated]\label{theo:bayesnet:closeness}
There exists an algorithm (\cref{alg:bayesnet}) that, given sample access to two Bayesian networks $p$ and $q$ over $\{0, 1\}^{n}$ with in-degree at most $d$, as well as a parameter $\varepsilon\in(0,1]$, takes
\[
s = \tilde{O} \left( \min \left\{ \frac{2^{3 d / 4} n}{\varepsilon^2},
   \frac{2^{2 d / 3} n^{4 / 3}}{\varepsilon^{8 / 3}} \right\} +
   \frac{n^2}{\varepsilon^4} \right) , 
\]
samples from $p,q$, and distinguishes with probability at least $2/3$ between (1) $p=q$ and (2) $d_{\rm TV}(p,q) \geq \varepsilon$. Moreover, the algorithm runs in time $\tilde{O}(n^{d+1} s)$.
\end{theorem}

\begin{algorithm}[ht]
\caption{Closeness testing for Bayesian networks}
\label{alg:bayesnet}
\begin{algorithmic}[1]
\Require Sample access to $p$ and $q$, both in-degree $d$ Bayesian networks over $\{0,1\}^n$; accuracy parameter $\varepsilon \in (0,1]$.
    \State Set $\varepsilon_1 \gets \Theta\left(\frac{\varepsilon^2}{n}\right)$, $\varepsilon_2 \gets \Theta\left( \frac{\varepsilon^2}{d \cdot n \cdot \log (d n / \varepsilon)}
  \right)$, $\delta \gets \frac{1}{20\binom{n}{d+1}}$
  \State Set \[
  m \gets \tilde{O} \left( \min \left\{ \frac{2^{3 d / 4} \cdot n}{\varepsilon^2},
       \frac{2^{2 d / 3} \cdot n^{4 / 3}}{\varepsilon^{8 / 3}} \right\} \right)  +  O \left( \frac{d^3
       n^2 \cdot \log^2 (n / \varepsilon)}{\varepsilon^4} \right)
       \]
  \State Take a multiset set $S_p$ (resp. $S_q$) of $m$ samples from $\tilde{p}$ (resp. $\tilde{q}$), defined by
  \[
    \tilde{p} \coloneqq \left( 1 - \frac{C_3 \cdot \varepsilon^2}{d n \log(n/\varepsilon)} \right) p + \frac{C_3 \cdot \varepsilon^2}{d n \log(n/\varepsilon)} \cdot U
\quad \text{and} \quad\tilde{q} \coloneqq \left( 1 - \frac{C_3 \cdot \varepsilon^2}{d n \log(n/\varepsilon)} \right) q + \frac{C_3 \cdot \varepsilon^2}{d n \log(n/\varepsilon)} \cdot U
\]
  \ForAll{subset $T\subseteq [n]$ with $|T|=d+1$} \Comment{$\binom{n}{d + 1}$ possible subsets}
    \State Let $\mathcal{T} \gets \{0,1\}^T$
    \State Test $\tilde{p}_{\mathcal{T}} = \tilde{q}_{\mathcal{T}}$ vs. $| H
  (\tilde{p}_{\mathcal{T}}) - H (\tilde{q}_{\mathcal{T}}) | \geq \varepsilon_1$ with failure probability $\delta$, using $S_p,S_q$ \\ \Comment{By~\cref{theo:eet:full_informal}}
    \State Test $\tilde{p}_{\mathcal{T}} = \tilde{q}_{\mathcal{T}}$ vs.
  $d^2_H (\tilde{p}_{\mathcal{T}}, \tilde{q}_{\mathcal{T}}) \geq \varepsilon_2$ with failure probability $\delta$, using $S_p,S_q$ \Comment{By~\cite{DK16}}
  \If{either test rejects}
    \Return \textsf{reject}
  \EndIf
  \EndFor
  
  \State \Return \textsf{accept}
\end{algorithmic}
\end{algorithm}
\begin{proof}  The pseudocode is provided in~\cref{alg:bayesnet}. 
\ifnum\randomcr=1
Due to space constraints, we defer the proof of correctness to the full version~\cite{CPSY26arxiv}.
\else
Before arguing the correctness of the algorithm, we analyze its sample and time complexities.
    \paragraph*{Sample complexity.} The sample complexity comes from running the two local tests, the entropy difference test and Hellinger tests, once for each possible subset $[n]$ of $d+1$ coordinates, using the same set of samples for all tests. Each test is thus run over a discrete domain of size at most $2^{d+1}$. 
    
    For the entropy test, we set the accuracy to be $\varepsilon' \coloneqq \varepsilon^2/n$ and the error probability $\delta \coloneqq 1/(20n^{d+1})$ (so that via a union bound we can conclude that none of the tests are incorrect, with probability at least $19/20$), which leads to the following number of samples via \cref{theo:eet:full_informal}: 
    \begin{align*}
      \log (n^{d + 1})  \left( \tilde{O}  \left( \min \left\{ \frac{(2^{d + 1})^{3
      / 4}}{\varepsilon^2 / n}, \frac{(2^{d + 1})^{2 / 3}}{(\varepsilon^2 / n)^{4
      / 3}} \right\} \right) + O \left( \frac{\log^2 (2^{d + 1} / (\varepsilon^2 /
      n))}{(\varepsilon^2 / n)^2} \right) \right),
    \end{align*}
   which in turn simplifies to
    \begin{align*} \tilde{O} \left( \min \left\{ \frac{2^{3 d / 4} \cdot n}{\varepsilon^2},
       \frac{2^{2 d / 3} \cdot n^{4 / 3}}{\varepsilon^{8 / 3}} \right\} \right) + O \left( \frac{d^3
       n^2 \cdot \log^2 (n / \varepsilon) \cdot \log(n)}{\varepsilon^4} \right) . \end{align*}
    For the Hellinger test, we set the threshold to be $\varepsilon'' = O \left( \sqrt{\frac{\varepsilon^2}{d \cdot n \cdot \log (d n / \varepsilon)}}
    \right)$ and the same failure probability $\delta$, for a sample complexity of 
    \begin{align*}\tilde{O} \left( \min \left\{ \frac{2^{3 d / 4} n}{\varepsilon^2}, \frac{2^{2
    d / 3} n^{4 / 3}}{\varepsilon^{8 / 3}} \right\} \right) . \end{align*}
    The total sample complexity is then
    \begin{align*} \tilde{O} \left( \min \left\{ \frac{2^{3 d / 4} \cdot n}{\varepsilon^2},
       \frac{2^{2 d / 3} \cdot n^{4 / 3}}{\varepsilon^{8 / 3}} \right\} + \frac{
       n^2}{\varepsilon^4} \right),
       \end{align*}
    as stated. (Note that the exponential terms in $d$ allow $\tilde{O}(\cdot)$ notation to hide ${\rm poly}(d)$ terms.)

   \paragraph*{Running time.} The algorithm runs over all $O(n^{d+1})$ subsets of size $d+1$, and for each conducts two local tests: an entropy difference test (\cref{theo:eet:full_informal}) and a Hellinger closeness test. The total runtime is dominated by $O(n^{d+1})$ times the running times of these two subroutines, each of which runs in time $O(m \log m)$ on its sample budget $m$, for a total runtime of
\begin{align*}
\tilde{O}\!\left( n^{d+1} \cdot \left( \min\!\left\{ \frac{2^{3d/4} \cdot n}{\varepsilon^2}, \frac{2^{2d/3} \cdot n^{4/3}}{\varepsilon^{8/3}} \right\} + \frac{n^2}{\varepsilon^4} \right) \right)\,.
\end{align*}
    \paragraph*{Correctness part 1 (completeness).} We first argue about the completeness of the algorithm, before turning to the more challenging step of establishing soundness. Assume $p=q$. Then $\tilde{p}=\tilde{q}$, and $\tilde{p}_{\mathcal{T}} = \tilde{q}_{\mathcal{T}}$ for every $\mathcal{T}$: by a union bound over all $2\cdot \binom{n}{d+1}$ tests performed, all successfully accept with probability at least $\frac{9}{10}$, in which case the algorithm returns $\textsf{accept}$.

    \paragraph*{Correctness part 2 (soundness).} Suppose $d_{\rm TV}(p,q) \geq \varepsilon$. In what follows, we will assume that all $2\cdot \binom{n}{d+1}$ behaved correctly, which by a union bound occurs with probability at least $9/10$.

        We will argue that then either (1) at least one of the \EET~tests must return $\textsf{reject}$, or  (2) at least one of the Hellinger tests must return $\textsf{reject}$. To do so, assume (1) does not hold, i.e., that all $\binom{n}{d+1}$ \EET~tests returned $\textsf{accept}$. 
        Our starting point is the following lemma from~\cite{CY24}:
    \begin{lemma}[{\cite[Lemma~3.4]{CY24}}]
        \label{lem:cy24:kl}
        Let $\varepsilon\in(0,1/2]$, and $p$ and $q$ be two Bayes nets with in-degree at most $d$ supported on $\{0,1\}^n$ such that, for each subset $T \subseteq [n]$ of $d+1$ coordinates, letting $\mathcal{T} \coloneqq \{0,1\}^T$,
        \[
            |H(p_{\mathcal{T}}) - H(q_{\mathcal{T}}) | \leq \frac{\varepsilon^2}{n}\,.
        \]
        Suppose further that $p$ is Markov w.r.t. $G$ and $q$  is Markov w.r.t. $G'$. Then
        \[
        d_{\tmop{KL}} (p \| p_{G'}) \leq O
(\varepsilon^2) \quad\text{ and }\quad d_{\tmop{KL}} (q \| q_{G}) \leq O
(\varepsilon^2)\,.
        \]
    \end{lemma}
    We are interested in applying this lemma to $G$, an in-degree $d$ DAG such that $p$ is Markov with respect to $G$. Indeed, as all \EET~tests from~\cref{alg:bayesnet} correctly returned $\textsf{accept}$, we must have $|H(p_{\mathcal{T}}) - H(q_{\mathcal{T}}) | < \varepsilon_1 = \frac{\varepsilon^2}{n}$ for all $\mathcal{T}$, which Lemma \ref{lem:cy24:kl} implies $d_{\tmop{KL}} (q\| q_G) \leq O
(\varepsilon^2)$ (and, of course, $d_{\tmop{KL}} (p\| p_G) = 0$). We \emph{would like} to use this (and the fact that $d_{\rm TV}(p,q) \geq \varepsilon$) to argue that 
$d_{\rm KL} (p_{\mathcal{T}} \| q_{\mathcal{T}})$ must then be noticeably large for \emph{some} $\mathcal{T}$, and from this conclude that $d^2_H (p_{\mathcal{T}}, q_{\mathcal{T}})$ must be noticeably large for that same $\mathcal{T}$ (causing the corresponding Hellinger test to reject). Unfortunately, this approach would fail, as the last step (relating KL and squared Hellinger in this direction) is not possible in general, if the probabilities of $p,q$ can be arbitrarily small. 

To alleviate this issue, we use as a proxy in our algorithm a suitable mixture of $p,q$ with the uniform distribution (over $\{0, 1\}^n$):  specifically, let $\tilde{p} \coloneqq \left( 1 - C_3 
\frac{\varepsilon^2}{d n \log(n/\varepsilon)} \right) p + C_3 \frac{\varepsilon^2}{d n \log(n/\varepsilon)} \cdot U$, and
$\tilde{q} \coloneqq \left( 1 - C_3 \frac{\varepsilon^2}{d n \log(n/\varepsilon)} \right) q + C_3
\frac{\varepsilon^2}{d n \log(n/\varepsilon)} \cdot U$. Now, if we marginalize over any subset, e.g., a coordinate $X_i$ and its (at most) $d$ parents $\Pi_i$, then we will
have
\begin{equation}
  \sum_{x_1, \ldots, x_n \backslash \{ x_i, \pi_i \} \in \{ 0, 1 \}^{n - |
  \Pi_i | - 1}} \tilde{p} (x) \geq 2^{n - d - 1} C_3 \cdot \frac{\varepsilon^2}{d n \log(n/\varepsilon)}
  \cdot \frac{1}{2^n} = C_3 \frac{\varepsilon^2}{2^{d + 1} d n \log(n/\varepsilon)},
\end{equation}
(and similarly for $\tilde{q}$), a lower bound which will allow us to lower bound Hellinger distance in terms of KL distance, using \cite[Lemma 3.1]{SST25}. The challenge now is
that $\tilde{p}$ is no longer necessarily Markov w.r.t.~a sparse DAG. Nevertheless, we show below that being \emph{close} to a distribution which is Markov w.r.t.~to a sparse DAG is enough for the argument to go through:
\begin{lemma}
  \label{lemma:mixture_preserve_approx_DAG}
  Fix any graph $G$ and distribution $p$, and a constant $C_3 > 0$, and 
  let $\tilde{p} = \left( 1 - \frac{C_3 \varepsilon^2}{d n \log(n/\varepsilon)} \right) \cdot p +
  \frac{C_3 \varepsilon^2}{d n \log(n/\varepsilon)} \cdot U$, then we have $| d_{\tmop{KL}} (\tilde{p}
  \| \tilde{p}_G) - d_{\tmop{KL}} (p\|p_G) | \leq {O}
  (\varepsilon^2)$.
\end{lemma}

\begin{proof}
  We first
  show that $| H (p) - H (\tilde{p}) | \leq O (\varepsilon^2)$ and $| H
  (p_{X_i, \Pi_i}) - H (\tilde{p}_{X_i, \Pi_i}) | \leq O \left(
  \frac{\varepsilon^2}{n} \right)$. Firstly, we know that $d_{\tmop{TV}} (p,
  \tilde{p}) \leq O \left( \frac{\varepsilon^2}{d n \log(n/\varepsilon)} \right)$ and so $| H
  (p) - H (\tilde{p}) | \leq d_{\tmop{TV}} (p, \tilde{p}) \cdot \log
  \left( \frac{2^n}{d_{\tmop{TV}} (p, \tilde{p})} \right) \leq {O}
  (\varepsilon^2)$. On the other hand, we can consider the marginalization
  calculation of $\tilde{p}$ on $X_i, \Pi_i$ as follows:
  \begin{align*}
    \tilde{p} (x_i, \pi_i) & = \sum_{x_1, \ldots, x_n \backslash \{ x_i,
    \pi_i \} \in \{ 0, 1 \}^{n - | \Pi_i | - 1}} \left( 1 - C_3
    \frac{\varepsilon^2}{d n \log(n/\varepsilon)} \right) \cdot p(x) + C_3 \frac{\varepsilon^2}{d n \log(n/\varepsilon)} \cdot
    U(x)\\
    & = \left( \sum_{x_1, \ldots, x_n \backslash \{ x_i, \pi_i \} \in \{ 0,
    1 \}^{n - | \Pi_i | - 1}} \left( 1 - C_3 \frac{\varepsilon^2}{d n \log(n/\varepsilon)} \right) \cdot
    p(x) \right) + C_3 \frac{\varepsilon^2}{d n \log(n/\varepsilon)} \cdot \frac{1}{2^{| \Pi_i | + 1}} .\\
    & = \left( 1 - C_3 \frac{\varepsilon^2}{d n \log(n/\varepsilon)} \right) \cdot p (x_i, \pi_i) +
    C_3 \frac{\varepsilon^2}{d n \log(n/\varepsilon)} \cdot \frac{1}{2^{| \Pi_i | + 1}} .
  \end{align*}
  This implies that $d_{\tmop{TV}} (\tilde{p}_{X_i, \Pi_i}, p_{X_i, \Pi_i})
  \leq O \left( \frac{\varepsilon^2}{d n \log(n/\varepsilon)} \right)$, and thus
  \begin{align*} | H (p_{X_i, \Pi_i}) - H (\tilde{p}_{X_i, \Pi_i}) | \leq O \left(
     \frac{\varepsilon^2}{d n \log(n/\varepsilon)} \cdot \log\left(\frac{ 2^{d + 1} d n \log(n/\varepsilon)}{
     \varepsilon^2 }\right) \right) = {O} \left(
     \frac{\varepsilon^2}{n} \right). 
\end{align*}
  By the classic Chow-Liu decomposition \cite{ChowL68} (and noting the relation between entropies and mutual information), it follows that
  \begin{align*}
    d_{\tmop{KL}} (\tilde{p} \| \tilde{p}_G) & = - \sum_{i = 1}^n I
    (\tilde{p}_{X_i} ; \tilde{p}_{\Pi_i}) + \sum_{i = 1}^n H (\tilde{p}_{X_i})
    - H (\tilde{p}_{X_1, \ldots, X_n})\\
    & = - \left( \sum_{i = 1}^n I (p_{X_i} ; p_{\Pi_i}) \pm {O}
    \left( \frac{\varepsilon^2}{n} \right) \right) + \left( \sum_{i = 1}^n H
    (p_{X_i}) \pm {O} \left( \frac{\varepsilon^2}{n} \right) \right) -
    (H (p_{X_1, \ldots, X_n}) \pm {O} (\varepsilon^2))\\
    & = d_{\tmop{KL}} (p\|p_G) \pm {O} (\varepsilon^2) . \qedhere
  \end{align*}
\end{proof} 
Similarly, the same holds for $q$ and (any) graph $G$: $| d_{\tmop{KL}} (\tilde{q} \|
\tilde{q}_G) - d_{\tmop{KL}} (q\|q_G) | \leq {O} (\varepsilon^2)$.

We now argue that there exists a DAG $G$ such that $d_{\tmop{KL}} (\tilde{p}
\| \tilde{p}_G) \leq {O} (\varepsilon^2)$ and $d_{\tmop{KL}}
(\tilde{p} \| \tilde{q}_G) \geq \Omega (\varepsilon^2)$. 
Indeed, from the discussion after~\cref{lem:cy24:kl}, we know that the DAG $G$ that $p$ is Markov with respect to 
will satisfy $d_{\tmop{KL}} (q\|q_G) \leq O (\varepsilon^2)$, and so we have
\begin{align} d_{\tmop{KL}} (\tilde{q} \| \tilde{q}_G) \leq O (\varepsilon^2)
   \quad\infixand\quad d_{\tmop{KL}} (\tilde{p} \| \tilde{p}_G) \leq O
   (\varepsilon^2),  \label{eq:bn:useful1}
\end{align}
by~\cref{lemma:mixture_preserve_approx_DAG}. Hence, by Pinsker's inequality,
$d_{\tmop{TV}} (\tilde{q}, \tilde{q}_G) \leq \sqrt{\frac{1}{2}d_{\tmop{KL}}
(\tilde{q} \| \tilde{q}_G)} \leq O (\varepsilon)$, where the constant in $O(\varepsilon)$ can be made sufficiently small by adjusting $C_3$ in \cref{lemma:mixture_preserve_approx_DAG}. Now, using
$d_{\tmop{TV}} (p, \tilde{p}), d_{\tmop{TV}} (q, \tilde{q}) \leq O \left(
\frac{\varepsilon^2}{n} \right)$ and the triangle inequality,
\begin{align*} \Omega (\varepsilon) \leq d_{\tmop{TV}} (p, q) \leq d_{\tmop{TV}}
   (p, \tilde{p}) + d_{\tmop{TV}} (\tilde{p}, \tilde{q}_G) + d_{\tmop{TV}}
   (\tilde{q}_G, \tilde{q}) + d_{\tmop{TV}} (\tilde{q}, q) . \end{align*}
Rearranging terms:
\begin{align}
  \sqrt{\frac{1}{2}d_{\tmop{KL}} (\tilde{p} \| \tilde{q}_G)} & \geq  d_{\tmop{TV}}
  (\tilde{p}, \tilde{q}_G) \notag\\
  & \geq  d_{\tmop{TV}} (p, q) - d_{\tmop{TV}} (p, \tilde{p}) -
  d_{\tmop{TV}} (\tilde{q}_G, \tilde{q}) - d_{\tmop{TV}} (\tilde{q}, q) \notag\\
  & \geq  \varepsilon - O \left( 2 \cdot
  \frac{\varepsilon^2}{n} \right) - O (\varepsilon) = \Omega (\varepsilon) . \label{eq:bn:useful2}
\end{align}
We will use the above two facts \eqref{eq:bn:useful1}, \eqref{eq:bn:useful2} to get towards our goal, which is to show that $d_{\rm KL} (p_{\mathcal{T}}\| q_{\mathcal{T}})$ must then be noticeably large (namely, $\Omega(\varepsilon^2)$) for \emph{some} $\mathcal{T}$ (of the form $\{X_i\}\cup \Pi_i$), in view of finally concluding the same for $d_{H} (p_{\mathcal{T}}, q_{\mathcal{T}})$. 
Squaring \eqref{eq:bn:useful2} and removing $d_{\tmop{KL}} (\tilde{p} \| \tilde{p}_G) \leq O
   (\varepsilon^2)$ from both sides, we get
\begin{align}
  \Omega (\varepsilon^2) & \leq \overbrace{d_{\tmop{KL}}
  (\tilde{p} \| \tilde{q}_G)}^{\Omega (\varepsilon^2)} - \overbrace{d_{\tmop{KL}} (\tilde{p} \| \tilde{p}_G)}^{O (\varepsilon^2)} \nonumber\\
  & = \sum_x \tilde{p} (x) \log \frac{\tilde{p} (x)}{\tilde{q}_G (x)} -
  \sum_x \tilde{p} (x) \log \frac{\tilde{p} (x)}{\tilde{p}_G (x)} \nonumber\\
  & = \sum_x \tilde{p} (x) \log \frac{\tilde{p}_G (x)}{\tilde{q}_G (x)} \nonumber\\
  & = \sum_x \tilde{p} (x) \log \frac{\prod_{i = 1}^n \tilde{p} (x_i |
  \pi_i)}{\prod_{i = 1}^n \tilde{q} (x_i | \pi_i)} \nonumber\\
  & = \sum_{i = 1}^n \sum_{x_i, \pi_i} \tilde{p} (x_i, \pi_i) \log
  \frac{\tilde{p} (x_i | \pi_i)}{\tilde{q} (x_i | \pi_i)} \nonumber\\
  & = \sum_{i = 1}^n \left( d_{\tmop{KL}} (\tilde{p}_{X_i, \Pi_i} \|
  \tilde{q}_{X_i, \Pi_i}) - d_{\tmop{KL}} (\tilde{p}_{\Pi_i} \|
  \tilde{q}_{\Pi_i}) \right) \nonumber\\
  &\leq \sum_{i = 1}^n d_{\tmop{KL}} (\tilde{p}_{X_i, \Pi_i} \| \tilde{q}_{X_i,
   \Pi_i}) \label{eq:local_KL_divergence_large} ~~~~~~~~~~~~~~~~~~~~~~~~~~~~~~\text{(since $d_{\tmop{KL}} (\tilde{p}_{\Pi_i} \| \tilde{q}_{\Pi_i}) \geq 0$)}
\end{align}
Now, since $d_{\tmop{TV}} (p, q)
\geq \Omega (\varepsilon)$, by using Pinsker's inequality and properties of KL divergence, an averaging argument shows that there must exist $i \in [n]$ such
that,
\begin{align*} d_{\tmop{KL}} (\tilde{p}_{X_i, \Pi_i} \| \tilde{q}_{X_i, \Pi_i}) \geq
   \Omega \left( \frac{\varepsilon^2}{n} \right) . \end{align*}
For this specific $i$, the relationship between $d_{\tmop{KL}}$ and $d_H^2$ in~\cite[Lemma 3.1]{SST25} gives
\begin{align*} d_{\tmop{KL}} (\tilde{p}_{X_i, \Pi_i} \| \tilde{q}_{X_i, \Pi_i}) \leq
   \left( 2 + \log \left( \max_{x \in \{ 0, 1 \}^{| \Pi_i | + 1}}
   \frac{\tilde{p}_{X_i, \Pi_i} (x)}{\tilde{q}_{X_i, \Pi_i} (x)} \right) \cdot
   d_H^2 (\tilde{p}_{X_i, \Pi_i}, \tilde{q}_{X_i, \Pi_i}) \right) . \end{align*}
Using the lower bound $\tilde{q}_{X_i, \Pi_i} (x) \geq
\frac{C_3 \varepsilon^2}{2^{d + 1} d n \log(n/\varepsilon)}$ (which follows from the mixture we took), we have that
\begin{align*} d_H^2 (\tilde{p}_{X_i, \Pi_i}, \tilde{q}_{X_i, \Pi_i}) \geq \Omega
   \left( \frac{\varepsilon^2}{n} \cdot \frac{1}{2 + \log \left( \frac{2^{d +
   1} d n \log(n/\varepsilon)}{C_3 \varepsilon^2} \right)} \right) = \Omega \left(
   \frac{\varepsilon^2}{d \cdot n \cdot \log \left( \frac{d n}{\varepsilon}
   \right)} \right) , \end{align*}
and so the Hellinger test, when considering a subset $\mathcal{T} \supseteq \{X_i\}\cup \Pi_i$, will output \textsf{reject}.

Having established both completeness and soundness, this concludes the proof.
\fi
\end{proof}
\begin{remark}
    We have an interesting dichotomy: our algorithm is efficient (polynomial in $n$) only if $d$ is constant. However, it is worse than its learning counterpart when $d$ is constant: $n^2$ (our test) vs. $n$ (learning). On the other hand, if $d = \Omega(\log n)$, then our tester will be more efficient than learning sample-wise, but its runtime becomes $n^{\Omega(\log n)}$, and is thus quasi-polynomial.
\end{remark}

As identity testing is a special case of closeness testing, our above analysis can straightforwardly be ported to yield an identity testing algorithm (using an algorithm for identity testing in Hellinger instead of a closeness one as the first building block). An immediate consequence of this is to yield an improved identity tester for Bayes net with sample complexity $\tilde{O}(\frac{2^{d/2}n}{\varepsilon^2}+\frac{n^2}{\varepsilon^4})$, thus improving on the prior bound $\tilde{O}(\frac{2^{d/2}n^{3/2}}{\varepsilon^2}+\frac{n^2}{\varepsilon^4})$ \cite[Theorem 1.2]{CY24}.

\begin{corollary}
    \label{cor:improve:upon:cy24}
    Given sample access to a Bayesian networks $p$ and a full description of
    Bayesian network $q$ over $\{0, 1\}^n$ with in-degree at most $d$, as well as
    a parameter $\varepsilon \in (0, 1]$, there is an algorithm that takes
    \[ s = \tilde{O} \left( \frac{2^{d / 2} n}{\varepsilon^2} +
       \frac{n^2}{\varepsilon^4} \right) , \]
    samples from $p$, and distinguishes with probability at least $2 / 3$ between
    (1) $p = q$ and (2) $d_{\mathrm{TV}} (p, q) \geq \varepsilon$. Moreover, the
    algorithm runs in time $\tilde{O} (n^{d + 1} s)$.
\end{corollary}

%% file: lowerbound.tex
\section{Lower Bound}
\label{sec:lower-bound}

In this section, we establish our sample complexity lower bound for \EET~(\cref{thm:eet-lb:intro}). The
$\widetilde\Omega(\min\{n^{3/4}/\varepsilon,\, n^{2/3}/\varepsilon^{4/3}\})$
term follows from a reduction to mutual information (\MI) testing
combined with \citet[Corollary~6.2]{SST25}, and the
$\Omega(\log^2 n / \varepsilon^2)$ term follows directly from
\citet[Lemma~2.7]{CY24}, whose two-point construction is also a hard instance
for \EET.

For a joint distribution $P_{AC}$ over
$[k_A] \times [k_C]$ with marginals $P_A, P_C$, the \emph{mutual information}
between coordinates is
\begin{align*}
  I(A : C) := H(A) + H(C) - H(A, C),
\end{align*}
and equals zero if and only if $P_{AC} = P_A \otimes P_C$. The \MI~testing problem asks, given sample access to $P_{AC}$, to distinguish
$I(A : C) = 0$ from $I(A : C) \geq \varepsilon$ with probability at least $2/3$.
\citet[Corollary~6.2]{SST25} established the lower bound
\begin{align*}
  \widetilde\Omega\!\left(\min\!\left\{
    \frac{k_A^{3/4} k_C^{1/4}}{\varepsilon},\,
    \frac{k_A^{2/3} k_C^{1/3}}{\varepsilon^{4/3}}
  \right\}\right)
\end{align*}
for this task.

\paragraph*{Reduction from \MI~testing to \EET.}
\begin{lemma}
\label{lem:mi-to-eet}
Let $k_A, k_C \geq 2$ and $n = k_A k_C$. Any \EET~algorithm
$\mathcal{A}$ over $[n]$ with sample complexity $T$ yields an \MI~testing algorithm over $[k_A] \times [k_C]$ with sample complexity $3T$ and the
same success probability.
\end{lemma}

\begin{proof}
Map $[k_A] \times [k_C]$ with $[n]$ via some fixed bijection. Using
$\mathcal{A}$, we construct a \MI~tester for $P_{AC}$ over
$[k_A] \times [k_C]$ as follows:
\begin{enumerate}
  \item Draw $3T$ i.i.d.\ samples $(a_1, c_1), \ldots, (a_{3T}, c_{3T})$ from $P_{AC}$.
  \item Form $T$ samples of $P := P_{AC}$, denoted by $(a_i, c_i)$ for $i = 1, \ldots, T$.
  \item Form $T$ samples of $Q := P_A \otimes P_C$, denoted by $(a_{T+2i-1},\, c_{T+2i})$
        for $i = 1, \ldots, T$. Since the samples are i.i.d., the $a$- and $c$-coordinates
        of each pair are independent draws from $P_A$ and $P_C$ respectively.
  \item Run $\mathcal{A}$ on the $T$ samples of $P$ and the $T$ samples of $Q$.
\end{enumerate}
Observe that $|H(P) - H(Q)| = |H(A, C) - H(A) - H(C)| = I(A : C)$.  Then:
\begin{itemize}
  \item If $I(A : C) = 0$, then $P_{AC} = P_A \otimes P_C$, so $P = Q$ and
        $\mathcal{A}$ returns \textsf{accept}.
  \item If $I(A : C) \geq \varepsilon$, then $|H(P) - H(Q)| \geq \varepsilon$ and
        $\mathcal{A}$ returns \textsf{reject}.
\end{itemize}
\end{proof}

\begin{theorem}[{\EET} lower bound: \cref{thm:eet-lb:intro}, restated]
\label{thm:eet-lb}
Any algorithm that, given sample access to two unknown distributions $p, q$ over
$[n]$ and parameter $\varepsilon \in (0, 1]$, distinguishes $p = q$ from
$|H(p) - H(q)| \geq \varepsilon$ with probability at least $2/3$ requires
\begin{align*}
 \widetilde\Omega\!\left(\min\!\left\{\frac{n^{3/4}}{\varepsilon},\,
    \frac{n^{2/3}}{\varepsilon^{4/3}}\right\}\right)
  + \Omega\!\left(\frac{\log^2 n}{\varepsilon^2}\right)
\end{align*}
samples from each of $p$ and $q$.
\end{theorem}

\begin{proof}
Apply \cref{lem:mi-to-eet} with $k_A = n/2$ and $k_C = 2$. Combining with
\cite[Corollary~6.2]{SST25} gives the first term:
\begin{align*}
  \widetilde\Omega\!\left(\min\!\left\{
      \frac{(n/2)^{3/4} \cdot 2^{1/4}}{\varepsilon},\,
      \frac{(n/2)^{2/3} \cdot 2^{1/3}}{\varepsilon^{4/3}}
    \right\}\right)
  = \widetilde\Omega\!\left(\min\!\left\{\frac{n^{3/4}}{\varepsilon},\,
      \frac{n^{2/3}}{\varepsilon^{4/3}}\right\}\right).
\end{align*}
For the second term, \cite[Lemma~2.7]{CY24} exhibits distributions
$p, q$ over $[n]$ with $|H(p) - H(q)| \geq \varepsilon$
that no algorithm can distinguish from $p = q$ using fewer than
$\Omega(\log^2 n / \varepsilon^2)$ samples from $p$ when $q$ is
given. Any \EET~algorithm using $T$ samples from each of $p, q$
yields such a tester by simulating $T$ samples from the known $q$, so
\EET~also requires $\Omega(\log^2 n / \varepsilon^2)$ samples.
\end{proof}

%% file: Appendix.tex
\appendix
\section{Deferred Proofs}\label{app:deferred}
\factorialmoment*
\begin{proof}
We first show that $\mathbb{E} [X \cdot f (X)] = \lambda \cdot \mathbb{E} [f (X
  + 1)]$ by standard Poisson calculations:
    \begin{eqnarray*}
      &  & \mathbb{E} [X \cdot f (X)] = \sum_{k = 0} \frac{\lambda^k e^{-
      \lambda}}{k!} \cdot k \cdot f (k) = \lambda \cdot \sum_{k = 1}
      \frac{\lambda^{k - 1} e^{- \lambda}}{(k - 1) !} \cdot f (k) = \lambda \cdot
      \sum_{k = 0} \frac{\lambda^k e^{- \lambda}}{k!} f (k + 1)\\
      & = & \lambda \cdot \mathbb{E} [f (X + 1)]
    \end{eqnarray*}
  The proof is then completed by induction, by noting that
  \begin{align*}
    \mathbb{E} [(X)_m \cdot f (X)] & = \mathbb{E} [X \cdot (X-1)_{m - 1} \cdot
    f (X)]\\
    & = \lambda \cdot \mathbb{E} [(X-1+1)_{m - 1} \cdot f (X + 1)]\\
    & = \lambda^2 \cdot \mathbb{E} [(X)_{m - 2} \cdot f (X + 2)]\\
    & = \dots\\
    & = \lambda^m \cdot \mathbb{E} [f (X + m)] .\qedhere
  \end{align*}
\end{proof}